\documentclass[journal]{IEEEtran}

\usepackage[T1]{fontenc}
\usepackage[cmex10]{amsmath}

\DeclareMathOperator*{\argmin}{arg\,min}

\usepackage{algorithm}
\usepackage[noend]{algpseudocode}
\usepackage{pgfplots}
\makeatletter
\def\BState{\State\hskip-\ALG@thistlm}
\makeatother
\usepackage{algpseudocode}
\usepackage{cite}
\usepackage[utf8]{inputenc}
\usepackage{algorithm}
\usepackage{graphicx}
\usepackage{subcaption}
\usepackage{pstricks}
\usepackage{color}
\usepackage{bbm}
\usepackage{tcolorbox}
\usepackage{tikz}
\usepackage{pgfplots}
\usepackage{wrapfig}
\usepackage{lipsum}  % generates 
\usetikzlibrary{positioning}

\usepackage{amssymb}
\usepackage{amsthm}

\usepackage{relsize}
\usepackage{bbm}
\usepackage{bm}
\usepackage[font={small}]{caption}
\usepackage{comment,color,soul}
\usepackage{enumerate} 
\usetikzlibrary{arrows,automata}
\usepackage{amsfonts}

\usepackage{mathtools}
\usepackage{url}

\makeatletter
\def\thm@space@setup{\thm@preskip=2pt
        \thm@postskip=2pt \itshape}
\makeatother

\newtheoremstyle{newstyle}
{} %Aboveskip
{} %Below skip
{\mdseries} %Body font e.g.\mdseries,\bfseries,\scshape,\itshape
{} %Indent
{\bfseries} %Head font e.g.\bfseries,\scshape,\itshape
{.} %Punctuation afer theorem header
{ } %Space after theorem header
{} %Heading

\theoremstyle{newstyle}

\newtheorem{theorem}{Theorem}
\newtheorem{lemma}{Lemma}
\newtheorem{proposition}{Proposition}

\theoremstyle{definition}
\newtheorem*{example*}{Example}
\newtheorem{example}{Example}
\newtheorem{definition}{Definition}

\theoremstyle{remark}
\newtheorem{remark}{Remark}
\newtheorem{claim}{Claim}

\IEEEoverridecommandlockouts
\newcommand{\Expc}{\mathbb{E}}
\newcommand{\Prob}{\mathbb{P}}

\newcommand{\tA}{\tilde{A}}

\DeclarePairedDelimiter\floor{\lfloor}{\rfloor}

\newcommand{\verteq}{\rotatebox{90}{$\,=$}}

\begin{document}
 \sloppy

               \setlength{\belowcaptionskip}{-6pt}
        \setlength{\abovedisplayskip}{1mm}
        \setlength{\belowdisplayskip}{1mm}
        \setlength{\abovecaptionskip}{1mm}

        \title{Coded Computing for Distributed Graph Analytics} \author{Saurav Prakash$^*$, Amirhossein Reisizadeh$^*$, Ramtin Pedarsani, Amir Salman Avestimehr
        \thanks{$^*$Authors have equal contribution.}
        \thanks{S. Prakash and A. S. Avestimehr are with the the Department
        of Electrical and Computer Engineering, University of Southern California, Los Angeles, CA 90089 USA
        (e-mail: sauravpr@usc.edu; avestimehr@ee.usc.edu).}
        \thanks{A. Reisizadeh and R. Pedarsani are with the Department
        of Electrical and Computer Engineering, University of California, Santa Barbara, Santa Barbara, CA 93106 USA (e-mail: reisizadeh@ucsb.edu; ramtin@ece.ucsb.edu).}
        \thanks{We sincerely thank the editor and all the reviewers for their  valuable feedback and  detailed  comments. This work is supported by a startup grant for Ramtin Pedarsani and  NSF grants CCF-1408639, CCF-1755808, NETS-1419632, Office of Naval Research (ONR) award N000141612189, NSA grant, a research gift from Intel and by Defense Advanced Research Projects Agency (DARPA) under Contract No. HR001117C0053. The views, opinions, and/or findings expressed are those of the author(s) and should not be interpreted as representing the official views or policies of the Department of Defense or the U.S. Government.}
        \thanks{A part
        of this work was presented at IEEE International Symposium on Information
        Theory, 2018 \cite{saurav2018coding}.}
        
        \thanks{Copyright (c) 2017 IEEE. Personal use of this material is permitted.  However, permission to use this material for any other purposes must be obtained from the IEEE by sending a request to pubs-permissions@ieee.org.}}
        
        \maketitle

\begin{abstract} 
Many distributed computing systems have been developed recently for implementing graph based algorithms such as PageRank over large-scale graph-structured datasets such as social networks. Performance of these systems significantly suffers from \emph{communication bottleneck} as a large number of messages are exchanged among servers at each step of the computation. Motivated by graph based MapReduce, we propose a coded computing framework that leverages computation redundancy to alleviate the communication bottleneck in distributed graph processing. As a key contribution of this work, we develop a novel \emph{coding} scheme that systematically injects structured redundancy in the computation phase to enable \textit{coded} multicasting opportunities during message exchange between servers, reducing the communication load substantially in large-scale graph processing. For theoretical analysis, we consider random graph models, and focus on schemes in which subgraph allocation and Reduce allocation are only dependent on vertex ID while the Shuffle design varies with graph connectivity. Specifically, we prove that our proposed scheme enables an (asymptotically) inverse-linear trade-off between \textit{computation load} and \textit{average communication load} for two popular random graph models -- Erdös-Rényi model, and power law model. Particularly, for a given computation load $r$, (i.e. when each graph vertex is carefully stored at $r$ servers), the proposed scheme slashes the average communication load by (nearly) a multiplicative factor of $r$. Furthermore, for the Erdös-Rényi model, we prove that our proposed scheme is optimal asymptotically as the graph size increases by providing an information-theoretic converse. To illustrate the benefits of our scheme in practice, we implement PageRank over Amazon EC2, using artificial as well as real-world datasets, demonstrating gains of up to $50.8\%$ in comparison to the conventional PageRank implementation. Additionally, we specialize our coded scheme and extend our theoretical results to two other random graph models -- random bi-partite model, and stochastic block model. Our specialized schemes asymptotically enable inverse-linear trade-offs between computation and communication loads in distributed graph processing for these popular random graph models as well. We complement the achievability results with converse bounds for both of these models.
\end{abstract}


\IEEEpeerreviewmaketitle
\section{Introduction}
\IEEEPARstart{G}{raphs} are widely used to identify and incorporate the relationship patterns and anomalies inherent in real-life datasets. Their growing scale and importance have prompted the development of various large-scale distributed graph processing frameworks, such as Pregel~\cite{malewicz2010pregel}, PowerGraph~\cite{gonzalez2012powergraph} and GraphLab~\cite{low2012distributed}. The underlying theme in these systems is the \emph{think like a vertex} approach \cite{mccune2015thinking} where the computation at each vertex requires only the data available in the neighborhood of the vertex (see Fig. \ref{fig:intro1mapillustrate} for an illustrative example). This approach significantly improves performance in comparison to general-purpose distributed data processing systems (e.g., Dryad \cite{isard2007dryad}), which do not leverage the underlying structure of graphs. 

\begin{figure}[h!] 
    \centering
    \includegraphics[width=.47\textwidth]{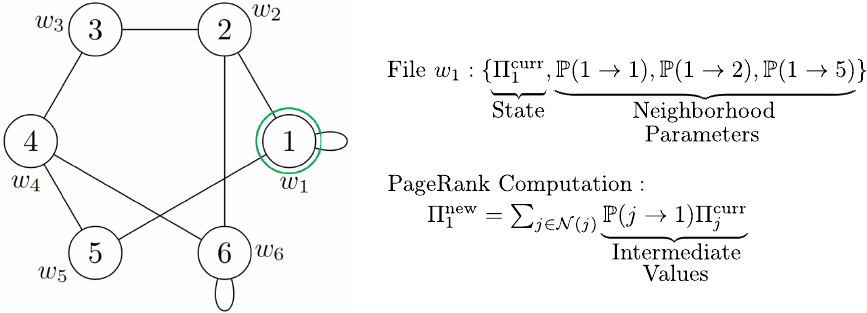}
    \caption{Illustrating the \emph{think like a vertex} paradigm prevalent in common parallel graph computing frameworks. The computation associated with a vertex only depends on its neighbors. In this example, we consider the PageRank computation over a graph with six vertices. Using vertex $1$ for representation, we illustrate the file and PageRank update at each vertex. File $w_1$ contains the state (current PageRank $\Pi_1^{\text{curr}}$) and the neighborhood parameters (probabilities of transitioning to neighbors $\{\mathbb{P}(1\rightarrow 1),\mathbb{P}(1\rightarrow 2),\mathbb{P}(1\rightarrow 5)\}$). The PageRank update associated with vertex $1$ is a function of only the neighborhood files (specifically, of the PageRanks of neighboring vertices and the transition probabilities from neighbors to vertex $1$).} 
    \label{fig:intro1mapillustrate}
\end{figure}

In these distributed graph processing systems, different subgraphs are stored at different servers, where a subgraph refers to the set of files associated with a subset of graph vertices. As a result of the distributed subgraph allocation, for carrying out the graph computation for a given vertex at a particular server, the intermediate values corresponding to the neighboring vertices whose files are not available at the server have to be communicated from other servers. These distributed graph processing systems, therefore, require many messages to be exchanged among servers during job execution. This results in communication bottleneck in parallel computations over graphs \cite{lumsdaine2007challenges}, accounting for more than $50 \%$ of the overall execution time in representative cases \cite{chen2014computation}. 

To alleviate the communication bottleneck in distributed graph processing, we develop a new framework that leverages \textit{computation redundancy} by computing the intermediate values at multiple servers via \textit{redundant subgraph allocation}. The redundancy in computation of intermediate values at multiple servers allows coded multicasting opportunities during exchange of messages between servers, thus reducing the communication load. Our proposed framework comprises of a mathematical model for MapReduce decomposition~\cite{dean2008mapreduce} of the graph computation task. The Map computation for a vertex corresponds to computing the intermediate values for the vertices in its neighborhood, while the Reduce computation for a vertex corresponds to combining the intermediate values from the neighboring vertices to obtain the final result of graph computation. Referring to the example in Fig. \ref{fig:intro1mapillustrate}, the Map and Reduce computations associated with vertex $1$ are as follows:
\begin{align}
    &\text{Map: }\Pi_1^{\text{curr}}\rightarrow \{v_{1,1}, v_{2,1}, v_{5,1}\},\nonumber\\
    &\text{Reduce: }\Pi_1^{\text{new}} = v_{1,1}+v_{1,2}+v_{1,5}, \nonumber
\end{align}
where $v_{j,i}=\mathbb{P}(i\rightarrow j)\Pi_i^{\text{curr}}$ is the intermediate value obtained from the Map computation of vertex $i\in\mathcal{N}(j)$. 

In distributed graph based MapReduce, each server is allocated a subgraph for Map computations and Reduce tasks for a subset of graph vertices, and the overall execution takes place in three phases -- Map, Shuffle, and Reduce. During Map phase, each server computes the intermediate values associated with the files in the allocated subgraph. During Shuffle phase, servers communicate with each other to exchange  missing intermediate values that are needed for executing the allocated Reduce tasks.  Finally, each server carries out the Reduce computations allocated to it to obtain the final results, using the intermediate values obtained locally during the Map phase and the missing intermediate values obtained from other servers during the Shuffle phase. Using our mathematical model of graph based MapReduce, our framework proposes to trade redundant computations in the Map phase with communication load during the Shuffle phase. The key idea is to leverage the graph structure and create coded messages during the Shuffle phase that simultaneously satisfy the data demand of multiple computing servers in the Reduce phase. 

Our work is rooted in the recent development of a coding framework that establishes an inverse-linear trade-off between computation and communication for general MapReduce computations -- Coded Distributed Computing (CDC) \cite{li2017fundamental}. In the MapReduce formulation considered in~\cite{li2017fundamental}, there are $n$ input files and the goal is to compute $Q$ output functions, where \emph{each} of the $Q$ output functions depends on \emph{all} of the $n$ input files. In CDC, each Map computation is carefully repeated at $r$ servers. The injected redundancy provides coded multicast opportunities in the Shuffle phase where servers exchange coded messages that are simultaneously useful for multiple servers. Each server then decodes the received messages and executes the Reduce computations assigned to it. Compared to uncoded Shuffle, where the required intermediate values are transmitted without leveraging coded multicast, CDC slashes the communication load by $r$. However, in contrast to graph based MapReduce considered in our framework, CDC does not incorporate the heterogeneity in the file requirements by the Reducers, as each Reducer in CDC is assumed to need intermediate values corresponding to all input files. 

\begin{figure}[h!] 
    \centering
    \includegraphics[width=.44\textwidth]{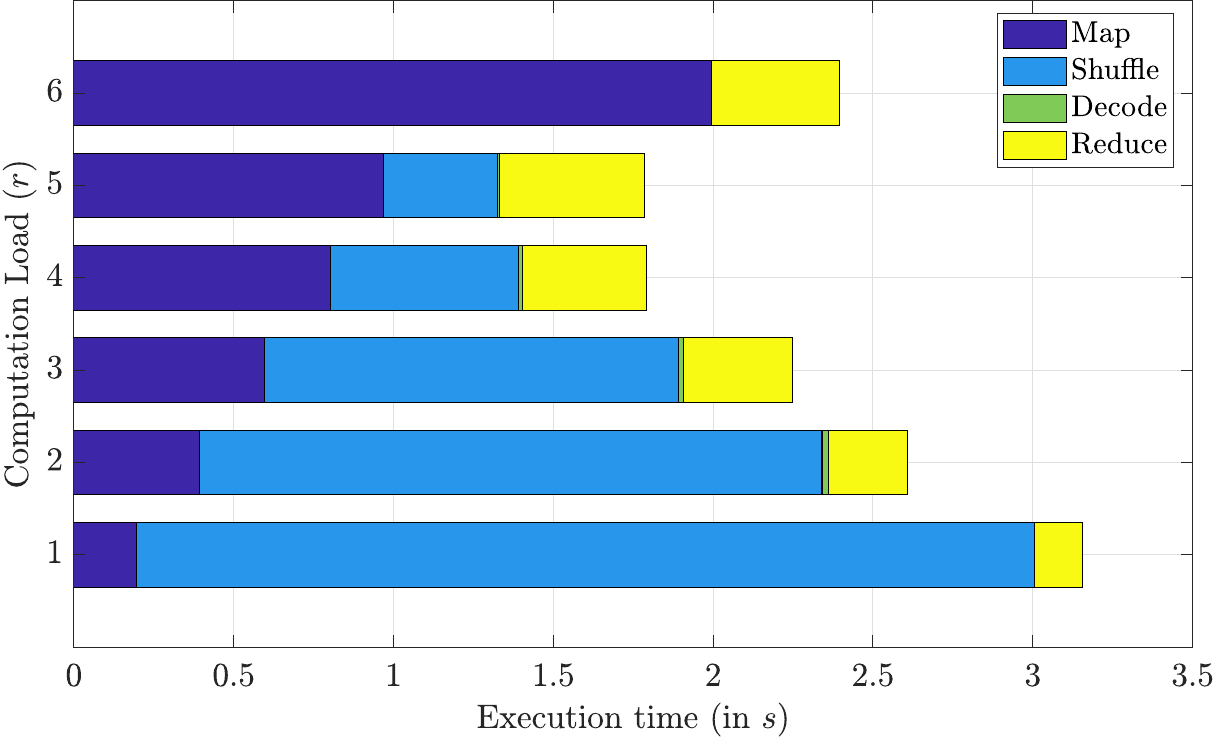}
    \caption{\footnotesize{Demonstrating the impact of our proposed coded scheme in practice. We consider PageRank implementation over a real-world dataset in an Amazon EC2 cluster consisting of $6$ servers. In this figure, we have illustrated the \textit{overall execution time} as well as the times spent in different phases of execution, as a function of computation load $r$ (details of implementation are provided in Section \ref{sec:expe}). One can observe that the Shuffle phase is the major component of the overall execution time in conventional PageRank implementation (computation load $r=1$), and our proposed coded scheme slashes the overall execution time by shortening the Shuffle phase (i.e., reducing the communication load) at the expense of increasing the Map phase (i.e. increasing the Map computations).}} 
    \label{fig:intro}
\end{figure}

Moving from the MapReduce framework in~\cite{li2017fundamental} to graph based MapReduce, the key challenge is that the computation associated with each vertex highly depends on the graph structure. In particular, graph computation at each vertex requires data \textit{only} from the neighboring vertices, while in the MapReduce framework in~\cite{li2017fundamental}, each output computation needs all the input files (which in graph based MapReduce shall correspond to a complete graph). This asymmetry in the data requirements of the graph computations is the main challenge in developing efficient subgraph and Reduce computation allocations and Shuffling schemes for graph based MapReduce. As a key component of our proposed coding framework, we propose a coded scheme that creates coding opportunities for communicating messages across servers by Mapping the same graph vertex at different servers, so that each coded transmission satisfies the data demand of multiple servers. Within each multicast group, each server communicates a coded message which is generated using careful alignment of the intermediate values that the server needs to communicate to all the remaining members of the multicast group. Each server retrieves the missing intermediate values required for its Reduce computations using the locally available intermediate values from the Map phase and the coded messages received during the Shuffle phase. 

\iffalse

Partitioning of graphs in popular graph processing frameworks such as Pregel~\cite{malewicz2010pregel} is solely based on the vertex ID and not on the vertex neighborhood density. Furthermore, designing subgraph allocation, Reduce allocation and Shuffling schemes for characterizing the minimum communication load in (\ref{eg:compcommGSample}) is NP-hard in general. This is because for the case of computation load $r=1$, finding the minimum communication load is equivalent to finding the minimum $K$-cut over the graph, which is NP-hard for general graphs~\cite{dahlhaus1992complexity}. Additionally, existing heuristics for load balancing in distributed graph processing involve \textit{additional} steps such as migration of vertex files \textit{during} graph algorithm execution~\cite{khayyat2013mizan}, which adds latency to the overall execution time. Hence, we focus on the problem of finding the subgraph and Reduce allocation tuple $A\in \mathcal{A}(r)$ that minimizes the \textit{average} normalized communication load across all graph realizations $G$ of $\mathcal{G}$. 
\fi

For characterization of the performance of our proposed coding framework for distributed graph analytics, we focus on random undirected graph models. In popular graph processing frameworks such as Pregel~\cite{malewicz2010pregel}, the graph partitioning for distributed processing among a set of servers is solely based on vertex ID, such as using $\mathsf{hash}(\text{ID})\text{ mod }{K}$, where $K$ is the number of servers. Therefore, in our problem formulation, for a given computation load $r$ and a random graph $\mathcal{G}=(\mathcal{V},\mathcal{E})$, we focus on subgraph and Reduce computation allocations $\mathcal{A}(r)$ that are based only on vertex IDs and not on graph connectivity. Here, $\mathcal{V}$ and $\mathcal{E}$ respectively denote the vertex set and edge set of $\mathcal{G}$. Although the Map and Reduce allocations are functions solely of vertex IDs, the Shuffle design needs to incorporate the graph connectivity of the graph realizations so that the communication load is minimized. This motivates us to consider the characterization of the minimum average normalized communication load $L^*(r)$, which is defined as follows:
\begin{equation}\label{eg:compcomm1}
L^*(r)\coloneqq \inf_{A\in \mathcal{A}(r)}\Expc_{\mathcal{G}}[L_A(r,\mathcal{G})], \nonumber
\end{equation}
where $L_A(r,G)$ denotes the minimum normalized communication load for a realization $G$ of $\mathcal{G}$ for a given subgraph and Reduce allocation tuple $A\in \mathcal{A}(r)$. The normalization is with respect to the total size of all the intermediate values corresponding to a fully connected graph with same number of vertices. Further details are deferred to Section \ref{sec:setting}, where we describe our problem formulation in detail.

For two popular random graph models, Erdös-Rényi model and power law model, we prove that our proposed coded scheme asymptotically achieves an inverse-linear trade-off between computation load in the Map phase and average normalized communication load in the Shuffle phase. Furthermore, for the Erdös-Rényi model, we develop an information-theoretic converse for the average communication load given a computation load of $r$. Using the asymptotic achievability result, we prove that the converse for the Erdös-Rényi model is asymptotically tight, thus proving the asymptotic optimality of our proposed coded scheme. Specifically, for a given computation load $r$, we show that the minimum average normalized communication load is as follows: 
\begin{equation}
\label{eq:intro_main}
    L^*(r)\approx\frac{1}{r}p\left(1-\frac{r}{K}\right), \nonumber
\end{equation}
where $p$ is the edge probability in the Erdös-Rényi model of size $n$, and $K$ denotes the number of servers.

To illustrate the benefits of our proposed coded scheme in practice, we demonstrate via simulation results that even for the Erdös-Rényi model with finite $n$, our proposed coded scheme achieves an average communication load which is within a small gap from the information-theoretic lower bound. Furthermore, it provides a gain of (almost) $r$ in comparison to the baseline scheme with uncoded Shuffling. Additionally, we implement the PageRank algorithm over Amazon EC2 servers using artificial as well as real-world graphs, demonstrating how our proposed coded scheme can be applied in practice. Fig. \ref{fig:intro} illustrates the results of our experiments over the conventional PageRank approach ($r=1$) for a social network webgraph Marker Cafe Dataset~\cite{markercafe}. As demonstrated in Fig. \ref{fig:intro}, our proposed coded scheme achieves a speedup of up to $43.4\%$ over the conventional PageRank implementation and a speedup of $25.5\%$ over the single server implementation. The details of the implementation are provided in Section \ref{sec:expe}.

We also specialize our coded scheme and extend the achievability results to two additional random graph models, random bi-partite model and stochastic block model. Specifically, we leverage the community structure in these models to adapt our proposed scheme to these models. In the random bi-partite model, we observe that there are no intra-cluster edges, due to which intermediate values for a particular Reducer in one cluster only comes from Mappers in the other cluster. Therefore, we specialize our proposed coded scheme from Section \ref{sec:achvblty} for the random bi-partite model, partitioning the available servers in proportion to the cluster sizes, so that there is maximum overlap between Reducers corresponding to vertices in one cluster and Mappers corresponding to vertices in other cluster. Similarly, for the stochastic block model, we specialize our proposed coded scheme based on the observation that Reducers corresponding to vertices in one cluster depend on the Mappers corresponding to the vertices within the cluster with one probability (due to intra-cluster edges), and on the vertices in the other cluster with another probability (due to cross-cluster edges). 

For both the random bi-partite model and the stochastic block model, we provide converse bounds. For the random bi-partite model, we remove vertices (and the edges corresponding to them) from the larger cluster so that the reduced graph has two clusters of equal sizes. The reduced graph model thus has two sets of Mappers and Reducers, which correspond to two different Erdös-Rényi models. Applying our converse  bound for the Erdös-Rényi model, we arrive at the converse of the random bi-partite model. For the stochastic block model converse, the key idea is to randomly remove edges from the graph such that a larger Erdös-Rényi graph is obtained, then utilize a coupling argument, and finally use our information theoretic converse bound for the Erdös-Rényi model. Therefore, the modified coded schemes for these models demonstrate that inverse-linear trade-offs between computation and communication loads in distributed graph processing exists for these graph models as well.

\vspace{0.1in}
\noindent
\textbf{Related Work.} A number of coding theoretic strategies have been recently proposed to mitigate the bottlenecks in large scale distributed computing \cite{lee2016speeding,li2017fundamental}. Several generalizations to the Coded Distributed Computing (CDC) technique proposed in \cite{li2017fundamental} have been developed. The authors in \cite{li2017scalable} extend CDC to wireless scenarios. The work in \cite{li2016coded} extends CDC to multistage dataflows. An alternative trade-off between communication and distributed computation has been explored in \cite{ezzeldin2017communication} for MapReduce framework under predetermined storage constraints. Coding using resolvable designs has been proposed in \cite{konstantinidis2018leveraging}. \cite{kiamari2017heterogeneous} extends CDC to heterogeneous computing environments. The work in \cite{li2018compressed} proposes coding scheme for reducing communication load for computations associated with linear aggregation of intermediate results in the final Reduce stage. The key difference between our framework and each of these works is that general MapReduce computations over graphs have heterogeneity in the data requirements for the Reduce functions associated with the vertices. Other notable works that deal with communication bottleneck in distributed computation include \cite{attia2016information,attia2018near,chungubershuffle}, where the authors propose techniques to reduce communication load in data shuffling in distributed learning.

Apart from communication bottleneck, various coding theoretic works have been proposed to tackle the straggler bottleneck \cite{lee2017coded,lee2016speeding,dutta2016short,reisizadeh2017latency,tandon2016gradient,yu2017polynomial,lee2017high,li2017near,yang2018coded,wang2018fundamental, mallick2018rateless, wang2018coded,charles2017approximate,ye2018communication,li2016unified,yang2016fault,karakus2017straggler,avesti2018resource,gauri2018stale,majum2018ldpc,ferdinand2017anytime,keshtkarjahromi2018coded,shakya2018distributed}. Stragglers are slow processors that have significantly larger delay for completing their computational task, thus slowing down the overall job execution in distributed computation. The first paper in this line of research proposed erasure correcting codes for straggler mitigation in linear computation \cite{lee2016speeding}. The work in \cite{lee2017coded} explores the potential of the multicore nature of computing servers, while \cite{reisizadeh2017latency} 
extends the straggler mitigation for the matrix vector problem in wireless scenarios. Redundant short dot products for matrix multiplication with long vector has been proposed in \cite{dutta2016short}. The authors in \cite{reisizadeh2017latency} propose Heterogeneous Coded Matrix Multiplication (HCMM) scheme for matrix-vector multiplication in heterogeneous scenarios. In \cite{tandon2016gradient}, the authors propose gradient coding schemes for straggler mitigation in distributed batch gradient descent. Works in \cite{yu2017polynomial} and \cite{lee2017high} develop coding schemes for computing high-dimensional matrix-matrix
multiplication. A Coupon Collector based straggler mitigation scheme for batched gradient descent has been proposed in \cite{li2017near}. Other notable schemes include Substitute decoding for coded iterative computing \cite{yang2018coded}, coding for sparse matrix multiplications \cite{wang2018fundamental, mallick2018rateless, wang2018coded}, approximate gradient coding \cite{charles2017approximate},  
efficient gradient computation tackling both straggler and communication load \cite{ye2018communication}, a unified coding scheme for distributed matrix multiplication \cite{li2016unified}, logistic regression with unreliable components \cite{yang2016fault}, among others.

\noindent
\textbf{Notation.} We denote by $[n]$ the set $\{1,2,\ldots,n\}$ for $n \in \mathbb{N}$. For non-negative functions $f$ and $g$ of $n$, we denote $f = \Theta(g)$ if there are positive constants $c_1$, $c_2$ and $n_0 \in \mathbb{N}$ such that $c_1 \leq f(n)/g(n) \leq c_2$ for every $n \geq n_0$, and $f=o(g)$ if $f(n)/g(n)$ converges to $0$ as $n$ goes to infinity. We define $f = \omega(g)$, if for any positive constant $c$, there exists
a constant $n_0 \in \mathbb{N}$ such that $f(n) > c \cdot g(n)$ for every $n \geq n_0$. To ease the notation, we let $2 \times Bern(p)$ denote a random variable that takes on the value $2$ w.p. $p$ and $0$ otherwise.

\section{Problem Setting}
\label{sec:setting}

We now describe the setting and formulate our distributed graph analytics problem. In particular, we specify our computation model, distributed implementation model and our problem formulation based on random graphs. 

\subsection{Computation Model}\label{sec:compmodel}
We consider an undirected graph $\mathcal{G=(V,E)}$ where $\mathcal{V}=[n]$ and $\mathcal{E}=\{(i,j):i,j\in \mathcal{V}\}$ denote the set of graph vertices and the set of edges respectively. A binary file $w_i\in \mathbb{F}_{2^F}$ of size $F\in \mathbb{N}$ containing vertex state and neighborhood parameters is associated with each graph vertex $i\in \mathcal{V}$. We  denote by $\mathcal{W}=\{w_i:i\in \mathcal{V}\}$ the set of files associated with all vertices in the graph. The neighborhood of vertex $i$ is denoted by $\mathcal{N}(i)=\{j\in \mathcal{V}:(j,i)\in \mathcal{E}\}$ and the set of files in the neighborhood of $i$ is represented by $\mathcal{W}_{\mathcal{N}(i)}=\{w_j:j\in \mathcal{N}(i)\}$. In general, $\mathcal{G}$ can have self-loops, i.e., vertex $i$ can be contained in $\mathcal{N}(i)$. Furthermore, a graph computation is associated with each vertex $i\in \mathcal{V}$ as follows:
\begin{equation}
\label{eq:graph_comput}
    \phi_i:\mathbb{F}_{2^F}^{|\mathcal{N}(i)|}\rightarrow \mathbb{F}_{2^B}, \nonumber
\end{equation}
where $\phi_i(\cdot)$ is a function that maps the input files in $\mathcal{W}_{\mathcal{N}(i)}$ to a length $B$ binary stream $o_i=\phi_i(\mathcal{W}_{\mathcal{N}(i)})$. 

The computation $\phi_i(\cdot)$ can be represented as a MapReduce computation:
\begin{equation}
\label{eq:MapRed}
   \phi_i(\mathcal{W}_{\mathcal{N}(i)})=h_i(\{g_{i,j}(w_j):w_j\in \mathcal{W}_{\mathcal{N}(i)}\}), 
\end{equation}
where the Map function $g_{i,j}:\mathbb{F}_{2^F}\rightarrow\mathbb{F}_{2^T}$ Maps file $w_j$ to a length $T$ binary intermediate value  $v_{i,j}=g_{i,j}(w_j)$, $\forall i\in \mathcal{N}(j)$. The Reduce function $h_i:\mathbb{F}_{2^T}^{|\mathcal{N}(i)|}\rightarrow \mathbb{F}_{2^B}$ Reduces the intermediate values associated with the output function $\phi_i(\cdot)$ into the final output value $o_i=h_i(\{v_{i,j}:j\in \mathcal{N}(i)\})$.

We illustrate our computation model using the graph presented in the previous section. Fig. \ref{fig:ex}(\subref{fig:toyexmpl1}) illustrates the graph with $n=6$ vertices, where each vertex is associated with a file, while Fig. \ref{fig:ex}(\subref{fig:dec}) illustrates the  corresponding MapReduce computations.

Common graph based algorithms can be expressed in the MapReduce computation framework described above \cite{lin2010design}. For brevity, we present two popular graph algorithms and describe how they can be expressed in the proposed MapReduce computation framework.
\begin{example}
\label{ex:PageRank}
\textbf{PageRank} \cite{page1999pagerank,xing2004weighted} is a popular algorithm to measure the importance of the vertices in a webgraph based on the underlying hyperlink structure. In particular, the algorithm computes the likelihood that a random surfer would visit a page. Mathematically, the rank of a vertex $i$ satisfies the following relation:
\begin{equation}
    \Pi(i)=(1-d)\sum_{j\in \mathcal{N}(i)}\Pi(j)\Prob(j\to i)+d\frac{1}{|\mathcal{V}|}, \nonumber
\end{equation}
where $(1-d)$ is referred to as the damping factor, $\Pi(i)$ denotes the likelihood that the random surfer will arrive at vertex $i$, $|\mathcal{V}|$ is the total number of vertices in the webgraph, and $\Prob(j\to i)$ is the transition probability from vertex $j$ to vertex $i$. The graph computation can be carried out iteratively as follows: 
\begin{equation}
    \Pi^{k}(i)=(1-d)\sum_{j\in \mathcal{N}(i)}\Pi^{k-1}(j)\Prob(j\to i)+d\frac{1}{|\mathcal{V}|}, \nonumber
\end{equation}
where $k$ and $k-1$ are respectively the current and previous iterations and $\Pi^0(i)=\frac{1}{|\mathcal{V}|}$ for all $i\in \mathcal{V}$ and $k=1,2,\cdots$. The number of iterations depends on the stopping criterion for the algorithm. Usually, the algorithm is stopped when the change in the PageRank mass of each vertex is less than a pre-defined tolerance. The rank update at each vertex can be decomposed into Map and Reduce functions for each iteration $k$. For a given vertex $i$ and iteration $k$, let $w_i^k=\{\Pi^{k-1}(i)\}\cup\{\mathbb{P}(i\rightarrow j):j\in\mathcal{N}(i)\}$, and $\phi^{k}_{i}(\mathcal{W}^{k}_{\mathcal{N}(i)})=(1-d)\sum_{j\in \mathcal{N}(i)}\Pi^{k-1}(j)\Prob(j\to i)+d\frac{1}{|\mathcal{V}|}$. The Mapper $g_{i,j}(\cdot)$ maps file $w_j^k$ to the intermediate values $v^{k}_{i,j}=g_{i,j}(w^k_j)=\Pi^{k-1}(j)\Prob(j\to i)$ for all neighboring vertices $i\in \mathcal{N}(j)$. Using the intermediate values from the Map computations, the Reducer $h_i(\cdot)$ computes vertex $i$'s updated rank as $\Pi^{k}(i)=h_i\big(\{v^{k}_{i,j}:j \in \mathcal{N}(i)\}\big)=(1-d)\sum_{j\in \mathcal{N}(i)}  v^{k}_{i,j}+d\frac{1}{|\mathcal{V}|}$. 
\end{example}
\begin{example}
\textbf{Single-source shortest path} is one of the most studied problems in graph theory. The task here is to find the shortest path to each vertex $i$ in the graph from a source vertex $s$. A sub-problem for this task is to compute the distance of each vertex $i$ from the source vertex $s$, where distance $D(i)$ is the length of the shortest path from $s$ to $i$. This can be carried out iteratively in parallel. First, initialize $D^0(s)=0$ and $D^0(i)=+\infty, \forall i\in \mathcal{V}\setminus \{s\}$. Subsequently, each vertex $i$ is updated as follows at each iteration $k$:
\begin{equation}
    D^k(i)=\min_{j\in \mathcal{N}(i)} \left\{ D^{k-1}(j)+t(j,i) \right\}, \nonumber
\end{equation}
where $t(j,i)$ is the weight of the edge $(j,i)$. The algorithm is stopped when the change in the distance value for each vertex is within a pre-defined tolerance.
The distance computation for each vertex at iteration $k$ can be decomposed into Map and Reduce computations. Particularly, for each vertex $i$ and iteration $k$, let $w_i^k=\{D^{k-1}(i)\}\cup\{t(i,j):j\in\mathcal{N}(i)\}$, and $\phi_i^{k}(\mathcal{W}^k_{\mathcal{N}(i)})=\min_{j\in \mathcal{N}(i)} (D^{k-1}(j)+t(j,i))$. The Mapper $g_{i,j}(\cdot)$ Maps the file $w_j^k$ to the intermediate values $v^{k}_{i,j}=g_{i,j}(w^k_j)=D^{k-1}(j)+t(j,i)$ for all neighboring vertices $i\in \mathcal{N}(j)$. Using the intermediate values from the Map computations, the Reducer $h_i(\cdot)$ computes  $i$'s updated distance value as $D^{k}(i)=h_i \big(\{v^{k}_{i,j}:j \in \mathcal{N}(i)\}\big)=\min_{j \in \mathcal{N}(i)}v^{k}_{i,j}$. 
\end{example}

\subsection{Distributed Implementation}\label{sec:distimp}
For distributing the graph processing task, we consider $K$ servers that are  connected through a shared multicast network. Furthermore, at any given time, only one server can multicast over the shared network. Additionally, we assume that a multicast takes the same amount of time as a unicast. As described next, in order to distribute the Map computation tasks among the servers, each server is allocated a subgraph which is comprised of a subset of graph vertices and associated files that contain state and neighborhood information of vertices. 

\begin{figure}
\begin{subfigure}{\linewidth}
\centering
\includegraphics[width=.45\linewidth]{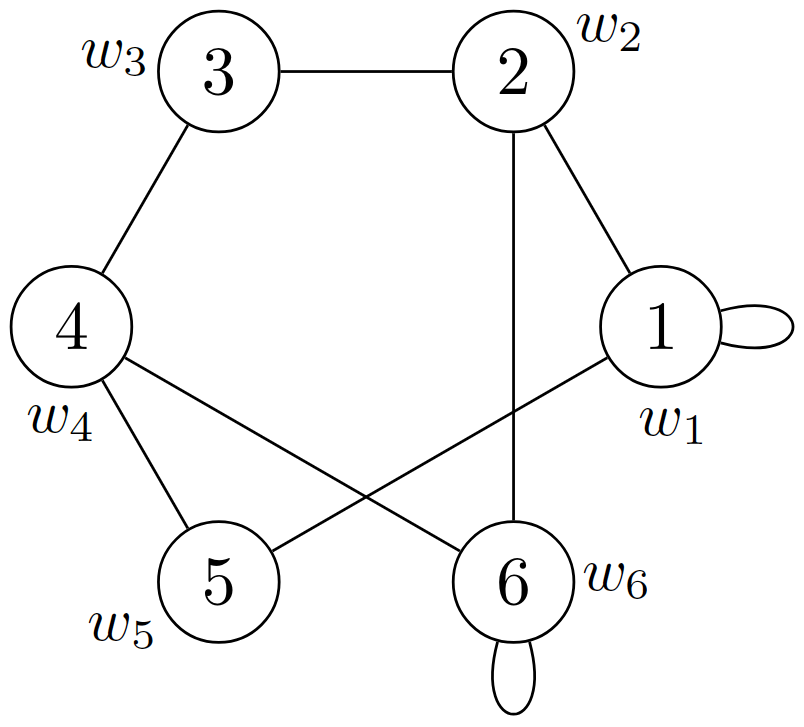}
\caption{An example of a graph with $6$ vertices, each of which has a file associated with it that contains its state and neighborhood parameters.}
\vspace{.4cm}
\label{fig:toyexmpl1}
\end{subfigure} \\
\begin{subfigure}{\linewidth}
\centering
\includegraphics[width=0.95\textwidth]{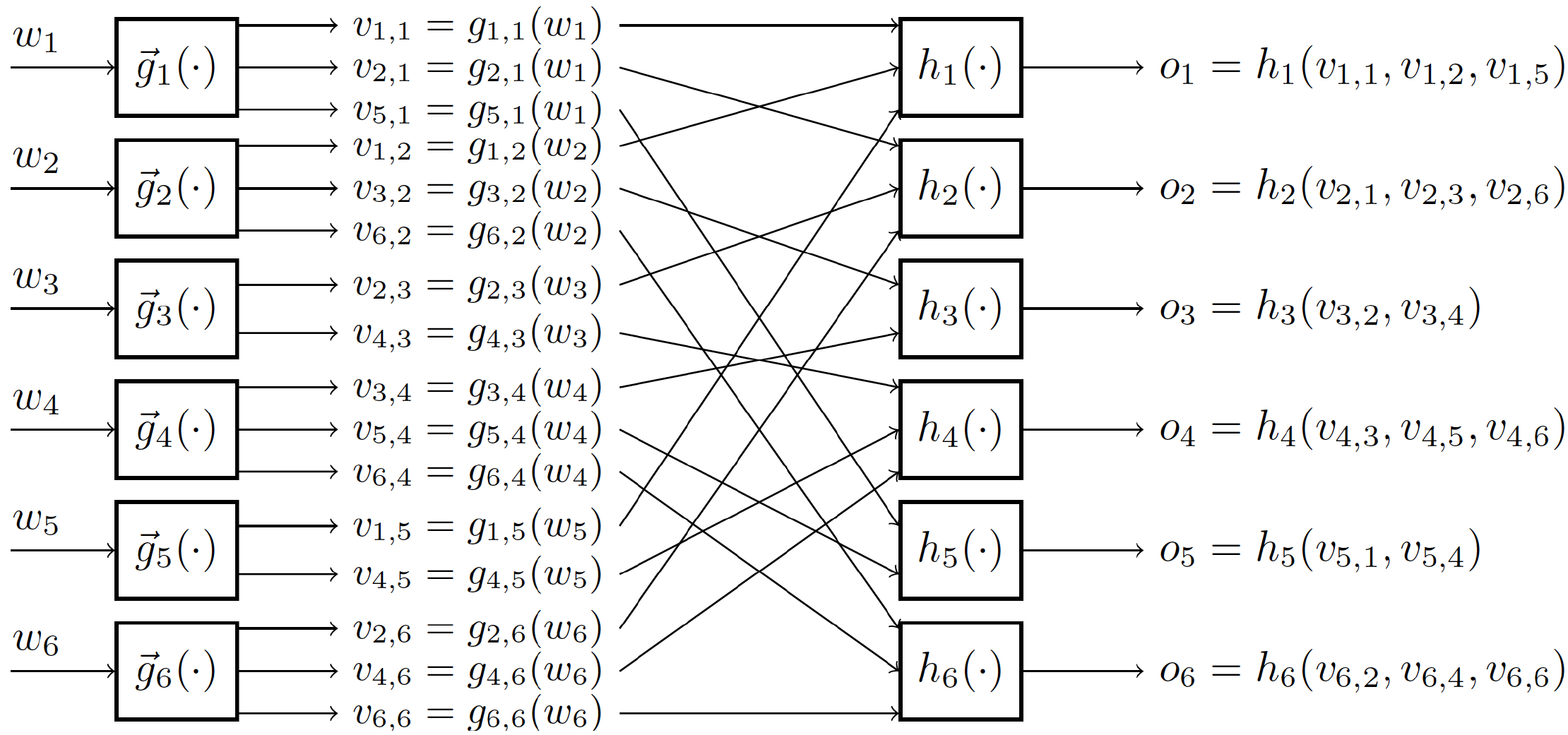}
\caption{MapReduce decomposition of the graph computations for the graph in Fig. \ref{fig:ex}(\subref{fig:toyexmpl1}).}
  \vspace{.5cm}
  \label{fig:dec}
\end{subfigure}\\[1ex]
\begin{subfigure}{1\linewidth}
\centering
\includegraphics[width=0.95\textwidth]{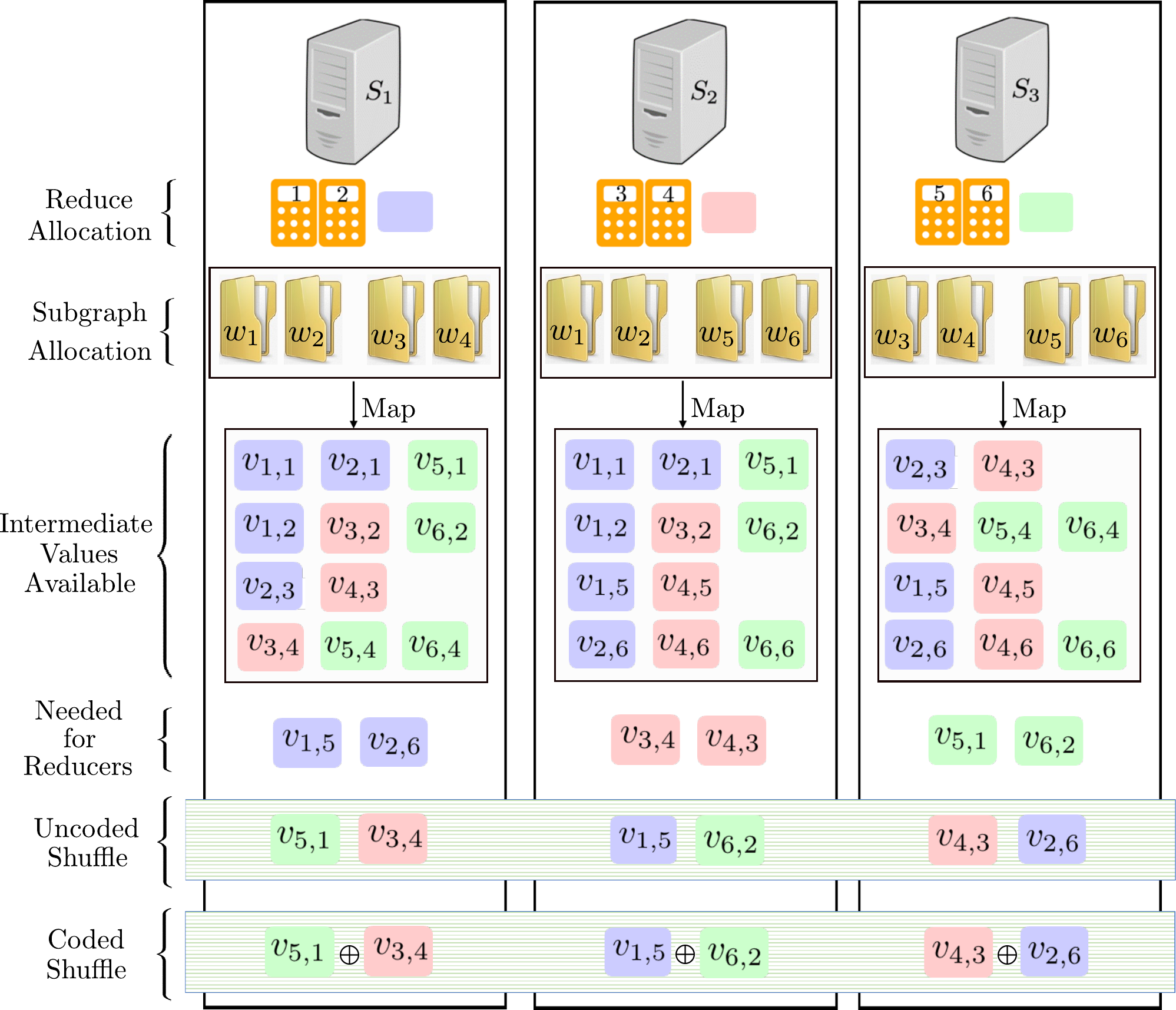}
\vspace{.1cm}
\caption{Illustration of subgraph and Reduce allocations for graph in Fig. \ref{fig:ex}(\subref{fig:toyexmpl1}) with computation load $r=2$ and $K=3$ servers. Each server is allocated a subgraph of size $4$ and $2$ Reducers. After the Map phase, each server needs to obtain the missing intermediate values that are needed to compute the Reduce functions allocated to it. Due to redundant subgraph allocation, each of the intermediate values missing at a server is available at both other servers. We illustrate two Shuffling schemes. In the uncoded Shuffle, a missing intermediate value is obtained from one of the other two servers, and each server is assigned the task of sending two intermediate values, one for each of the other two servers. In coded Shuffle, each server sends a XOR of the assigned intermediate values and sends \textit{only one} coded message which is simultaneously useful for the both other servers.}
  \vspace{.4cm}
  \label{fig:tabmapredex}
\end{subfigure}
\caption{An illustrative example.}
 \label{fig:ex}
  %\vspace{-.5cm}
\end{figure}

\noindent
\textbf{Subgraph Allocation}:
Each server is assigned the Map computations in (\ref{eq:MapRed}) associated with a subgraph, which consists of a subset of vertices and associated files containing state and neighborhood information of the vertices. We denote the subgraph that is allocated to each server $k\in [K]$ by $\mathcal{M}_k\subseteq \mathcal{V}$. Thus, server $k$ will then store all the files in $\mathcal{M}_k$, and will be responsible for computing the Map functions on those files. Note that each file should be Mapped by at least one server. Additionally, we allow redundant computations, i.e., each file can be Mapped by more than one server. The key idea in leveraging redundancy in the Map computation phase is to trade the computational resources in order to reduce the communication load in the Shuffle phase. We define the computation load as follows.
 
\begin{definition}[Computation Load]
\label{compload} For a subgraph allocation, $(\mathcal{M}_1,\cdots,\mathcal{M}_K)$, the computation load, $r\in[K]$, is defined as 
\begin{equation}
    r\coloneqq \frac{\sum_{k=1}^K|\mathcal{M}_k|}{n}, \nonumber
\end{equation}
where $|\mathcal{M}_k|$ denotes the number of vertices in the subgraph $\mathcal{M}_k$ for $k\in [K]$.
\end{definition}

\begin{remark}
For a desired computation load $r$, each server is assigned a subgraph with the same number of vertices, i.e. for each server $k\in [K]$, $|\mathcal{M}_k|=\frac{r n}{K}$. 
\end{remark}
To carry out the Reduce computation in (\ref{eq:MapRed}) for all vertices, each server is assigned a subset of Reduce functions as follows. 

\noindent
\textbf{Reduce Allocation}:
A Reducer is associated with each vertex of the graph $\mathcal{G}$ as represented in (\ref{eq:MapRed}). We use $\mathcal{R}_k\subseteq \mathcal{V}$ to denote the set of vertices whose Reduce computations are assigned to server $k\in [K]$. The set of Reduce computations are partitioned into $K$ equal parts and each part is associated exclusively with one server, i.e., $\cup_{k=1}^K \mathcal{R}_k=\mathcal{V}$ and $\mathcal{R}_m\cap \mathcal{R}_n=\phi$ for $m,n\in [K],m\neq n$. Therefore,  $|\mathcal{R}_k|=\frac{n}{K} $, $\forall k \in [K]$. 

For the graph in Fig. \ref{fig:ex}(\subref{fig:toyexmpl1}) and a computation load of $r=2$, we illustrate a scheme for subgraph allocation and Reduce allocation in Fig. \ref{fig:ex}(\subref{fig:tabmapredex}). Here, each vertex appears in exactly two subgraphs, i.e. Map computation associated with each vertex is assigned to exactly two servers. The subgraph and Reduce allocations in Fig. \ref{fig:ex}(\subref{fig:tabmapredex}) form key components of our proposed scheme in Section \ref{sec:achvblty}, in which for a computation load of $r$, every unique set of $r$ servers is allocated a unique batch of ${n}/{{K \choose r}}$ files for Map computations.

For a given scheme with subgraph allocation and Reduce allocation tuple denoted by  $A=(\mathcal{M},\mathcal{R})$, where $\mathcal{M}=(\mathcal{M}_1,\cdots,\mathcal{M}_K)$ and $\mathcal{R}=(\mathcal{R}_1,\cdots,\mathcal{R}_K)$, the distributed graph processing proceeds in three phases as described next.

\noindent
\textbf{Map phase}:
Each server first Maps the files associated with the subgraph that is allocated to it. More specifically, for each $i\in \mathcal{M}_k$, server $k$ computes a vector of intermediate values corresponding to the vertices in $\mathcal{N}(i)$ that is $\vec{g}_i=(v_{j,i}: j \in \mathcal{N}(i))$. For the running example, we illustrate the intermediate values generated at each server during the Map phase in Fig. \ref{fig:ex}(\subref{fig:tabmapredex}), where the color of an intermediate value denotes the server that is allocated the task to execute the corresponding Reducer. 

\noindent
\textbf{Shuffle phase:} To be able to do the final Reduce computations, each server needs the intermediate values corresponding to the neighbors of each vertex that it is responsible for its Reduction. Servers exchange messages so that at the end of the Shuffle phase, each server is able to recover its required set of intermediate values. More formally, the Shuffle phase proceeds as follows. For each $k\in [K]$,

\begin{enumerate}[(i)]
\item server $k$ creates a message $X_k \in \mathbb{F}_{2^{c_k}}$ as a function of intermediate values computed locally at that server during the Map phase, i.e. $X_k=\psi_k(\{ \vec{g}_i:i\in \mathcal{M}_k \})$, where $c_k$ is the length of the binary message $X_k$, 
\item server $k$ multicasts $X_k$ to all the remaining servers,
\item server $k$ recovers the missing intermediate values $\{v_{i,j}:i\in \mathcal{R}_k,  j\in \mathcal{N}(i), j\notin \mathcal{M}_k\}$ using locally computed intermediate values $\{v_{i,j}:i\in \mathcal{N}(j), j\in \mathcal{M}_k\}$ and received messages $ \{X_{k'}:k'\in [K]\setminus \{k\}\}$. 
\end{enumerate}

We define the normalized communication load of the Shuffle phase as follows.
\begin{definition}[Normalized Communication Load]
\label{commload} The normalized communication load, denoted by $L$, is defined as the number of bits communicated by $K$ servers during the Shuffle phase, normalized by the maximum possible total number of bits in the intermediate values associated with all the Reduce functions, i.e.
\begin{equation}
    L\coloneqq \frac{\sum_{k=1}^{K} c_k}{n^2T}. \nonumber
\end{equation}
\end{definition}
For the running example in Fig. \ref{fig:ex}(\subref{fig:tabmapredex}), after the Map phase, each server obtains the intermediate values corresponding to the files in its subgraph. The intermediate values that are needed for computing the allocated Reduce functions but are not available after the Map phase have also been highlighted. We illustrate an uncoded Shuffling scheme in which each server is assigned the task of sending some of its locally available intermediate values to other server over the shared multicast network. We highlight here that each intermediate value missing at a server is available at two other servers. For example, $v_{5,1}$ and $v_{6,2}$ are missing at server $3$, and both of them are available at servers $1$ and $2$. In this uncoded Shuffle, exactly one of the two servers is uniquely assigned the task to communicate the missing intermediate value to the server. For example, $v_{5,1}$ is multicasted by server $1$ while $v_{6,2}$ is multicasted by server $2$. As a total of $6$ intermediate values are sent over the shared multicast network, the normalized communication load of the uncoded Shuffle is $L=\frac{6}{36}$. 

The servers can instead send linear combinations of the intermediate values over the multicast network. For example, server $1$ multicasts $v_{5,1}\oplus v_{3,4}$. As $v_{5,1}$ is locally available at server $2$, server $2$ can compute $(v_{5,1}\oplus v_{3,4})\oplus v_{5,1}$ and obtain the missing intermediate value $v_{3,4}$. Similarly, server $3$ can obtain the missing intermediate value $v_{5,1}$. This illustrates that by using \textit{coded} Shuffle, in which each server sends a combination of locally available intermediate values over the multicast network, the communication load can be improved over the uncoded Shuffle. In this case specifically, the communication load for the coded Shuffle is $L=\frac{3}{36}$, which is factor of two (same as the computation load $r=2$) improvement over uncoded Shuffle. This forms the motivation behind our proposed scheme in Section \ref{sec:achvblty}. 

\noindent
\textbf{Reduce phase:}
 Using its locally computed intermediate values and the intermediate values recovered from the messages received from other servers during the Shuffle phase, server $k\in [K]$ computes the Reduce functions in $\mathcal{R}_k$ to calculate $o_i=h_i(\{v_{i,j}:j\in \mathcal{N}(i)\})$ for all $i\in \mathcal{R}_k$.

In Fig. \ref{fig:ex}(\subref{fig:tabmapredex}), each server has all the intermediate values that are needed to compute the allocated Reduce functions. For example, for computing the Reduce function associated with vertex $1$, server $1$ has intermediate values $v_{1,1}$ and $v_{1,2}$ available locally from the Map phase and the intermediate value $v_{1,5}$ obtained from server $2$ in the Shuffle phase. Therefore, each of the three servers can compute the Reduce functions allocated to it.  

\subsection{Problem Formulation}
As illustrated in Fig. \ref{fig:ex}, the communication load during Shuffle phase depends on subgraph allocation, Reduce allocation, and Shuffle strategy. For an allowed computation load $r$, our broader goal is to minimize the communication load during Shuffle phase through efficient schemes for allocation of subgraphs and Reducers to servers and for \textit{coded} Shuffling of intermediate values among the servers. We consider a random undirected graph $\mathcal{G=(V,E)}$, where edges independently exist with probability $\Prob[(i,j)\in \mathcal{E}]$ for all $i,j \in \mathcal{V}$. Let $\mathcal{A}(r)$ be the set of all possible subgraph and Reduce allocations for a given computation load $r$ (as defined in the previous subsection). For a graph realization $G$ of $\mathcal{G}$ and an allocation $A\in \mathcal{A}(r)$, a coded Shuffling scheme is feasible if each server can compute all the Reduce functions assigned to it. We denote by $L_{A}(r,G)$ the minimum (normalized) communication load (as defined in Definition \ref{commload}) over all feasible Shuffling coding schemes that enable each server  to compute all the Reduce functions assigned to it.\footnote{The uncoded Shuffling schemes are special cases of the coded Shuffling schemes and are thus included in the set of all feasible coded Shuffling schemes under consideration.} Hence, for a given realization $G$ of the random graph $\mathcal{G}$, the minimum communication load among all possible subgraph and Reduce allocations and feasible coded Shuffling schemes is as follows:
\begin{equation}
\label{eg:compcommGSample}
L_G^*(r)\coloneqq \inf_{A\in \mathcal{A}(r)}L_A(r,G).
\end{equation}
\begin{remark}
Partitioning of graphs in popular graph processing frameworks such as Pregel~\cite{malewicz2010pregel} is solely based on the vertex ID and not on the vertex neighborhood density. Furthermore, designing subgraph allocation, Reduce allocation and Shuffling schemes for characterizing the minimum communication load in (\ref{eg:compcommGSample}) is NP-hard in general. This is because for the case of computation load $r=1$, finding the minimum communication load is equivalent to finding the minimum $K$-cut over the graph, which is NP-hard for general graphs~\cite{dahlhaus1992complexity}. Additionally, existing heuristics for load balancing in distributed graph processing involve \textit{additional} steps such as migration of vertex files \textit{during} graph algorithm execution~\cite{khayyat2013mizan}, which adds latency to the overall execution time. Hence, we focus on the problem of finding the subgraph and Reduce allocation tuple $A\in \mathcal{A}(r)$ that minimizes the \textit{average} normalized communication load across all graph realizations $G$ of $\mathcal{G}$. 
\end{remark}

We formally define our problem as follows.

\noindent\textit{Problem:} For a given random undirected graph $\mathcal{G=(V,E)}$ and a computation load $r \in [K]$, our goal is to characterize the minimum average normalized communication load, i.e.
\begin{equation}\label{eg:compcomm}
L^*(r)\coloneqq \inf_{A\in \mathcal{A}(r)}\Expc_{\mathcal{G}}[L_A(r,\mathcal{G})].
\end{equation}

\begin{remark}
For $r \geq K$, $L^*(r)$ is trivially $0$ as each vertex can be mapped at each server, so all the intermediate values associated with the Reducers of any server is available at the server.
\end{remark}

\begin{remark}
As defined above, $L^*(r)$ essentially reveals a fundamental trade-off between computation and communication in distributed graph processing.
\end{remark}
\begin{remark}
In the above problem formulation, for a given subgraph and Reduce allocation tuple $A\in \mathcal{A}(r)$, in order to minimize the average communication load, the Shuffle scheme needs to take into consideration the connectivity of each realization $G$ of $\mathcal{G}$. As we describe in Section \ref{sec:achvblty}, our proposed coded scheme  utilizes careful alignment of intermediate values for creating coded messages for multicast during the Shuffle phase, leading to significant improvement in the average communication load.   
\end{remark}
\begin{remark}
Although the main focus of our problem formulation is on minimizing the average communication load for random graph models, our proposed coded scheme in Section \ref{sec:achvblty} is applicable to any real-world graph. As demonstrated in Section \ref{sec:expe}, our proposed coded scheme can provide significant performance gains in practice. Specifically, for implementing PageRank over the real-world social webgraph TheMarker Cafe~\cite{markercafe}, our proposed scheme provides a gain of up to $43.4\%$ in the \textit{overall execution time} in comparison to the conventional PageRank implementation. 
\end{remark}

In the next Section, we discuss our main results for four popular random graph models. 
\section{Main Results}

In this section, we present the main results of our work. Our first result is the characterization of $L^*(r)$ (defined in (\ref{eg:compcomm})) for the Erdös-Rényi model that is defined below.

\noindent
\textbf{Erdös-Rényi Model}: Denoted by $\text{ER}(n,p)$, this model consists of graphs of size $n$ in which each edge exists with probability $p\in(0,1]$, independently of other edges (Fig. \ref{fig:models}(\subref{fig:er})).

\begin{theorem}\label{thm:ERCDC}
For the Erdös-Rényi model ER$(n,p)$ with $p = \omega(\frac{1}{n^2})$, we have
\begin{equation}\label{eq:ER}
   \lim_{n\rightarrow \infty}\frac{L^*(r)}{p} = \frac{1}{r}\left(1-\frac{r}{K}\right). \nonumber
\end{equation}
\end{theorem}

\begin{remark}\label{remark:gain}
Theorem \ref{thm:ERCDC} reveals an interesting inverse-linear trade-off between computation and communication in distributed graph processing. Specifically, our proposed coded scheme in Section \ref{sec:achvblty} asymptotically gives a gain of $r$ in the average normalized communication load in comparison to the uncoded Shuffling scheme that as we discuss later in Section \ref{sec:achvblty}, only achieves an average normalized communication load of $p(1-\frac{r}{K})$. This trade-off can be used to leverage additional computing resources and capabilities to alleviate the costly communication bottleneck. Moreover, we numerically demonstrate that even for finite graphs, not only the proposed scheme significantly reduces the communication load in comparison to the uncoded scheme, but also has a small gap from the optimal average normalized communication load (Fig. \ref{fig:mainplot}). Finally, the assumption $p = \omega(\frac{1}{n^2})$ implies the regime of interest in  which the average number of edges in the graph is growing with $n$. Otherwise, the problem would not be of interest since the communication load would become negligible even without redundancy/coding in computation.
\end{remark}

\begin{figure}[ht]
\centering
\begin{tikzpicture}[scale=.94]
\begin{axis}[
  xlabel=Computation Load ($r$),
  ylabel=Expected Communication Load ($L$),
  legend style={nodes={scale=0.7}},
  label style={font=\small},
  legend cell align={left}]
\addplot table [x=r, y=uncoded]{data.dat};
\addlegendentry{Uncoded Scheme}
\addplot table [x=r, y=coded]{data.dat};
\addlegendentry{Proposed Coded Scheme}
\addplot table [x=r, y=bound]{data.dat};
\addlegendentry{Lower Bound}
\end{axis}
\end{tikzpicture}
\caption{Performance comparison of our proposed coded scheme with uncoded Shuffle scheme and the proposed lower bound. The averages for the communication load for the two schemes were obtained over graph realizations of the Erdös-Rényi model with $n=300$, $p=0.1 \text{ and } K=5$.}
\label{fig:mainplot}
\end{figure}

\begin{figure}[b!]
    \centering
    \begin{subfigure}[t]{0.45\linewidth}
        \centering \includegraphics[width=.8\linewidth]{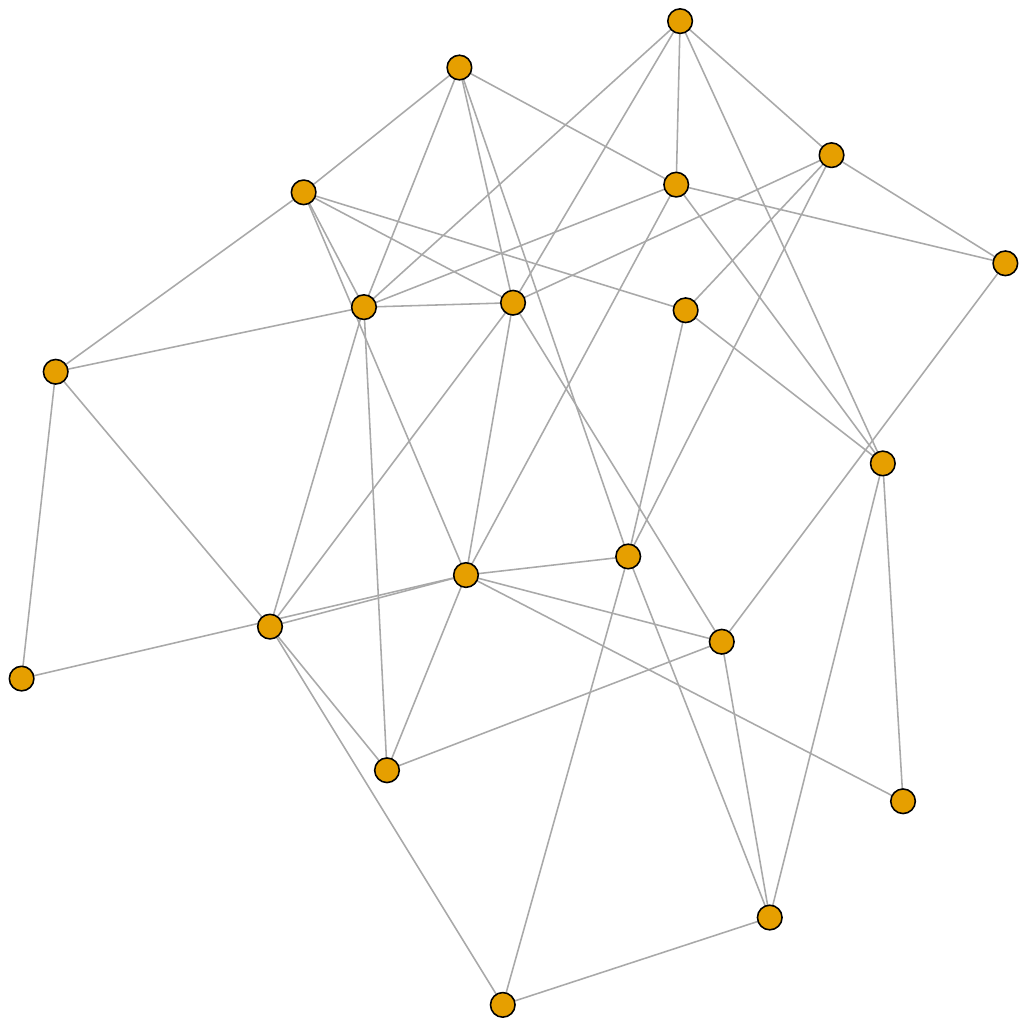} \caption{Erdös-Rényi model with $n= 20$.}\label{fig:er}
    \end{subfigure}%
    ~
    \begin{subfigure}[t]{0.45\linewidth}
        \centering \includegraphics[width=.8\linewidth]{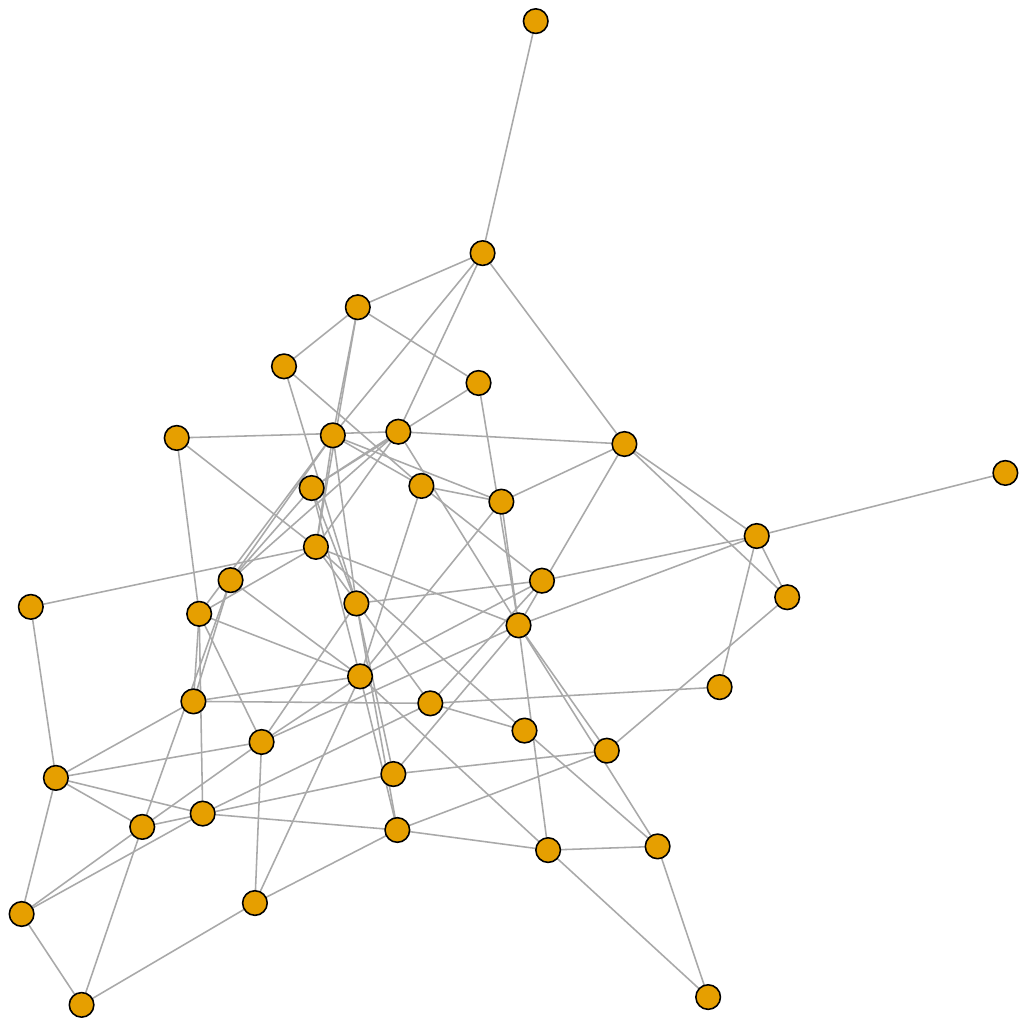}
        \caption{Power law model with $n=40$, $\gamma=2.3$ and $100$ edges.}\label{fig:pl}
    \end{subfigure} \\
    \vspace{.4cm}
    ~ 
    \begin{subfigure}[t]{0.45\linewidth}
        \centering \includegraphics[width=.8\linewidth]{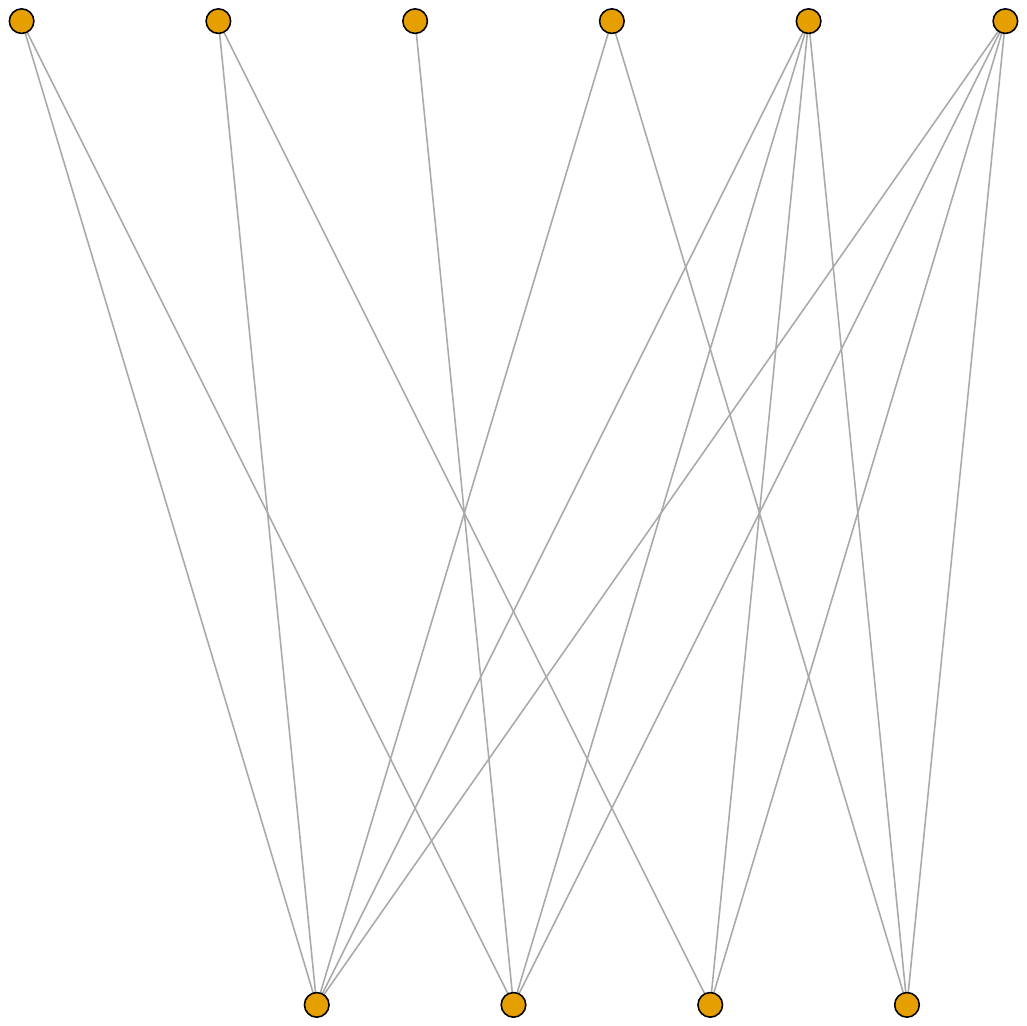}
        \caption{Random bipartite model with $n_1=6$ and $n_2=4$.}\label{fig:rb}
    \end{subfigure}
     ~ 
    \begin{subfigure}[t]{0.45\linewidth}
        \centering \includegraphics[width=.8\linewidth]{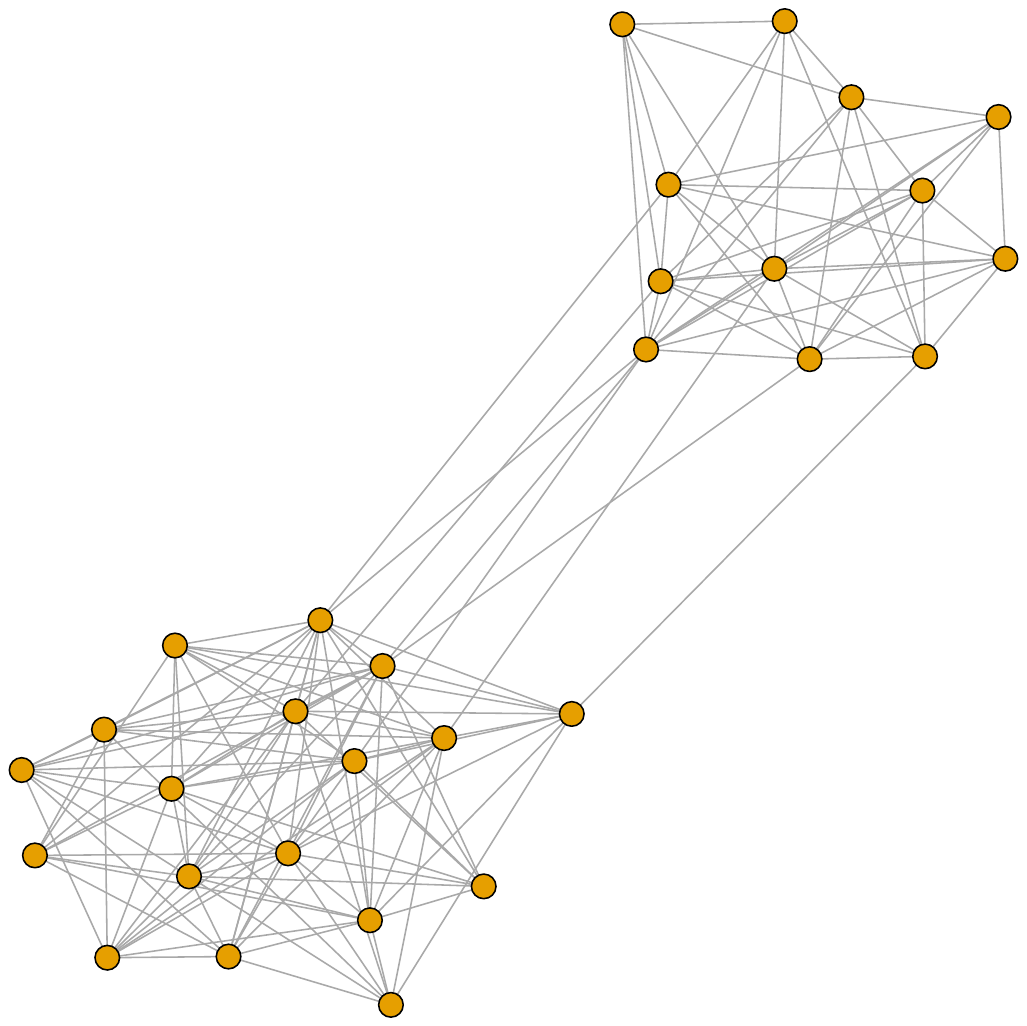}
        \caption{Stochastic block model with $n_1=12$ and $n_2=18$.}\label{fig:sbm}
    \end{subfigure}
     ~     
    \vspace{10pt}
    \caption{Illustrative instances of the random graph models considered in the paper. In Fig.  \ref{fig:models}(\subref{fig:er}), each edge exists with a given probability $p$. In Fig. \ref{fig:models}(\subref{fig:pl}), expected degree of each vertex follows a power law distribution with exponent $\gamma$. In Fig. \ref{fig:models}(\subref{fig:rb}), each cross-edge exists with a given probability $q$. In Fig. \ref{fig:models}(\subref{fig:sbm}), each intra-cluster edge exists with a given probability $p$ and each cross-edge exists with a given probability $q$.}\label{fig:models}
\end{figure}

\begin{remark}
Achievability Theorem \ref{thm:ERCDC} is proved in Section \ref{sec:achvblty}, where we provide subgraph and Reduce allocations followed by the code design for Shuffling for our proposed scheme. The main idea is to leverage the coded multicast opportunities offered by the injected redundancy and create coded messages which simultaneously satisfy the data demand of multiple servers. Careful combination of available intermediate values during the Shuffle phase benefits from the missing graph connections by aligning the intermediate values assigned to be communicated over the shared network. Conversely, Theorem \ref{thm:ERCDC} demonstrates that the asymptotic bandwidth gain $r$ achieved by the proposed scheme is optimal and can not be improved. For the proof of converse provided in  Section \ref{sec:ERconverse}, we use induction to derive information-theoretic lower bounds on the average normalized communication load required by any subset of servers and then use the induction on the set of all the $K$ servers.
\end{remark}

Our second result is the characterization of $L^*(r)$ for the power law model that is defined below.

\noindent
\textbf{Power Law Model}: Denoted by $\text{PL}(n,\gamma,\rho)$, this model consists of graphs of size $n$ in which degrees are i.i.d random variables drawn from a power law distribution with exponent $\gamma$ and edge probabilities are $\rho$-proportional to product of the degrees of the two end vertices (Fig. \ref{fig:models}(\subref{fig:pl})). 

\begin{theorem}
\label{thm:pl}
For the power law model graph $\text{PL}(n,\gamma,\rho)$ with node degrees $\{d_1,\cdots,d_n\}$, $\gamma>2$ and $\rho=\frac{1}{\sum_{i=1}^n d_i}$, we have
\begin{equation}
\label{eq:thm4}
     \limsup_{n \to \infty}\frac{n L^*(r)}{\left(\frac{\gamma-1}{\gamma-2} \right) } \leq \frac{1}{r}\left(1-\frac{r}{K}\right). \nonumber
 \end{equation}
\end{theorem}
\begin{remark}
Theorem \ref{thm:pl} demonstrates that an inverse-linear trade-off between computation load and communication load can also be achieved in the power law model. We leverage our coded scheme proposed in Section \ref{sec:achvblty} for the proof of Theorem \ref{thm:pl} in Section \ref{subsec:pl}.
\end{remark}

Furthermore, we specialize our proposed coded scheme in Section \ref{sec:achvblty} to develop subgraph allocation and Reduce allocation schemes along with coded Shuffling schemes for two other popular random graph models which are described below:

\noindent
\textbf{Random Bi-partite Model}: Denoted by $\text{RB}(n_1,n_2,q)$, this model consists of graphs with two disjoint clusters of sizes $n_1$ and $n_2$ in which each inter-cluster edge exists with probability $q\in (0,1]$, independently of other inter-cluster edges (Fig. \ref{fig:models}(\subref{fig:rb})). No intra-cluster edge exists in this model.

\noindent
\textbf{Stochastic Block Model}: Denoted by $\text{SBM}(n_1,n_2,p,q)$, this model consists of graphs with two disjoint clusters of sizes $n_1$ and $n_2$ such that each intra-cluster edge exists with probability $p$ and each inter-cluster edge exists with probability $q, 0 < q< p \leq 1$, all independent of each other (Fig. \ref{fig:models}(\subref{fig:sbm})).

The following theorems provide the achievability and converse results for RB and SBM models.

\begin{theorem}
\label{thm:rndmbiprt}
For the random bi-partite model $\text{RB}(n_1,n_2,q)$ with  $n=n_1+n_2$, $n_1=\Theta(n)$, $n_2=\Theta(n)$, $|n_1-n_2| = o(n)$ and $q = \omega(\frac{1}{n^2})$, we have
\begin{equation}
\label{eq:thm2}
    \frac{1}{8r}\left(1-\frac{2r}{K}\right)
    \leq
    \limsup_{n \to \infty}\frac{L^*(r)}{q} \leq \frac{1}{2r}\left(1-\frac{2r}{K}\right). \nonumber
\end{equation}
\end{theorem}

\begin{remark}
Theorem \ref{thm:rndmbiprt} characterizes the optimal average normalized communication load within a factor of $4$ for the random bi-partite model. We provide the proofs for achievability and converse of Theorem \ref{thm:rndmbiprt} in Appendices and \ref{subsec:RB} and  \ref{subsec:RB-converse} respectively. For achievability, we observe that there are no intra-cluster edges in the random bi-partite model, due to which intermediate values for a particular Reducer in one cluster only comes from Mappers in the other cluster. Therefore, we specialize our proposed coded scheme in Section \ref{sec:achvblty} for the random bi-partite model, partitioning the available servers in proportion to the cluster sizes. Therefore, there is maximum overlap between Reducers corresponding to vertices in one cluster and Mappers corresponding to vertices in other cluster. For proving the converse, we remove vertices (and the edges corresponding to them) from the larger cluster so that the reduced graph has two clusters of equal sizes. The reduced graph model thus has two sets of Mappers and Reducers, which correspond to two different Erdös-Rényi models. Applying our lower bound for the Erdös-Rényi model in Theorem \ref{thm:ERCDC}, we arrive at the converse of the bi-partite model. 
\end{remark}

\begin{theorem}
\label{thm:sbm}
For the stochastic block model $\text{SBM}(n_1,n_2,p,q)$ with $n=n_1+n_2$, $n_1=\Theta(n)$, $n_2=\Theta(n)$, and $p= \omega(\frac{1}{n^2}),q = \omega(\frac{1}{n^2})$, we have
\begin{equation}
\label{eq:thm3}
     \limsup_{n \to \infty}\frac{L^*(r)}{\frac{p n_1^2+p n_2^2+2q n_1n_2}{(n_1+n_2)^2} } \leq \frac{1}{r}\left(1-\frac{r}{K}\right). 
 \end{equation}
 Moreover, the following converse inequality holds:
 \begin{align}
 \label{eq:thm3conv}
    \frac{L^*(r)}{q} 
    & \geq
    \frac{1}{r} \left(1-\frac{r}{K}\right).
\end{align}
\end{theorem}
\begin{remark}
Using (\ref{eq:thm3}) and (\ref{eq:thm3conv}), it can be easily verified that for the stochastic block model, the converse is within a constant factor of achievability if $p=\Theta(q)$. The achievability and converse of Theorem \ref{thm:sbm} are proved in Appendices \ref{subsec:sbm} and \ref{subsec:SBM-converse} respectively. For achievability, we specialize our proposed coded scheme from Section \ref{sec:achvblty} based on the observation that in SBM, the Reducers corresponding to vertices in one cluster depend on the Mappers corresponding to the vertices within the cluster with one probability (due to intra-cluster edges), and on the vertices in the other cluster with another probability (due to cross-cluster edges). For the converse, the key idea is to randomly remove edges from the SBM model such that a larger ER model is obtained, then utilize a coupling argument, and finally use our information theoretic converse bound in Theorem \ref{thm:ERCDC}. 
\end{remark}

\section{Proposed Scheme and Proof of Achievability of Theorem \ref{thm:ERCDC}}
\label{sec:achvblty}
In this section, we first describe our proposed coded scheme for distributed graph analytics, and then leverage it to prove the achievability for the Erdös-Rényi model in Theorem \ref{thm:ERCDC}.

\subsection{Proposed Scheme}
\label{sec:GCDC}
%\vspace{-.1cm}
As described in our distributed graph processing framework in Section \ref{sec:setting}, a scheme for distributed implementation of the graph computations consists of subgraph allocation, Reduce allocation, and Shuffling algorithm. We next precisely describe our proposed scheme for a given realization $G$ of the underlying random graph $\mathcal{G}=(\mathcal{V},\mathcal{E})$.

\noindent \textbf{Subgraph Allocation:}
The $n$ files associated with the $n$ vertices of $G$ are first partitioned serially into ${K \choose r}$ batches $\mathcal{B}_1,\mathcal{B}_2,\ldots,\mathcal{B}_{{K \choose r}}$, where $\mathcal{B}_j$ comprises of the files associated with the vertices with IDs in the range $\{(j-1)g+1,(j-1)g+2,\ldots,jg\}$. Here, $g=n/{K \choose r}$ denotes the number of files in each batch. For our example with a graph of $6$ vertices, $3$ servers, and computation load $2$ presented in Section \ref{sec:setting}, the $6$ files are partitioned into ${K\choose r}=3$ batches each of size $g=2$ as follows (see Fig. \ref{fig:eg_scheme}(\subref{fig:eg_scheme_a})):
\begin{align}
    \mathcal{B}_1 = \{1,2\},\nonumber\\
    \mathcal{B}_2 = \{3,4\},\nonumber\\
    \mathcal{B}_3 = \{5,6\}.\nonumber
\end{align}

Each of the ${K \choose r}$ batches of files is associated with a unique set of $r$ servers. Specifically, let $\mathcal{F}_1,\mathcal{F}_2,\ldots,\mathcal{F}_{{K \choose r}}$ denote all possible combinations of the elements of $\{1,2,\ldots,K\}$. Then, each of the servers with indices in $\mathcal{F}_j$ is allocated each of the files contained in batch $\mathcal{B}_j$. 
Thus, server $k \in [K]$ Maps the vertices in $\mathcal{B}_{j}$ if $k\in \mathcal{F}_j$. Equivalently, $\mathcal{B}_j \subseteq \mathcal{M}_k$ if $k\in \mathcal{F}_j$. Therefore, we have the following for the subgraph allocation for server $k$: 
\begin{align}
\mathcal{M}_k=\cup_{j\in \left[{K \choose r}\right],k\in\mathcal{F}_j}\mathcal{B}_j. \nonumber
\end{align}
As each server is present in ${{K-1} \choose {r-1}}$ of the ${K \choose r}$ unique combinations of servers, we have the following for each server $k\in [K]$:
\begin{align}
    |\mathcal{M}_k|
    =
    {{K-1} \choose {r-1}}g
    =
    {{K-1} \choose {r-1}} \frac{n}{{{K} \choose {r}}} 
    =
    \frac{r n}{K}. \nonumber
\end{align}
In Fig. \ref{fig:eg_scheme}(\subref{fig:eg_scheme_a}), we illustrate the subgraph allocation for our running example. $\mathcal{F}_1=\{1,2\}$, $\mathcal{F}_2=\{1,3\}$ and $\mathcal{F}_3=\{2,3\}$. Each of the two files in batch $\mathcal{B}_j$ is assigned to each of the servers in $\mathcal{F}_j$, for $j\in\{1,2,3\}$. Thus, server $1$ is allocated files $\mathcal{B}_1\cup \mathcal{B}_2=\{w_1,w_2,w_3,w_4\}$, server $2$ is allocated  files $\mathcal{B}_1\cup \mathcal{B}_3=\{w_1,w_2,w_5,w_6\}$ and server $3$ is allocated $\mathcal{B}_2\cup \mathcal{B}_3=\{w_3,w_4,w_5,w_6\}$. Thus, $|\mathcal{M}_1|=|\mathcal{M}_2|=|\mathcal{M}_3|=4$.

\noindent \textbf{Reduce Allocation:} The $n$ Reduce functions associated with the $n$ graph vertices are disjointly and uniformly partitioned into $K$ subsets and each subset is assigned exclusively to one server. Specifically, for $k\in[K]$, $|\mathcal{R}_k|=\frac{n}{K}$ and  $\mathcal{R}_k=\{(k-1)\frac{n}{K}+1,(k-1)\frac{n}{K}+2,\ldots,k\frac{n}{K}\}$. In our running example, $\mathcal{R}_1=\{1,2\}$, $\mathcal{R}_2=\{3,4\}$ and $\mathcal{R}_3=\{5,6\}$. 

For notational convenience, we denote our proposed subgraph allocation and Reduce allocation by $A_{\textsf{C}}$. 

\begin{figure}[b!]
    \centering
    \begin{subfigure}[t]{\linewidth}
        \centering \includegraphics[width=.8\linewidth]{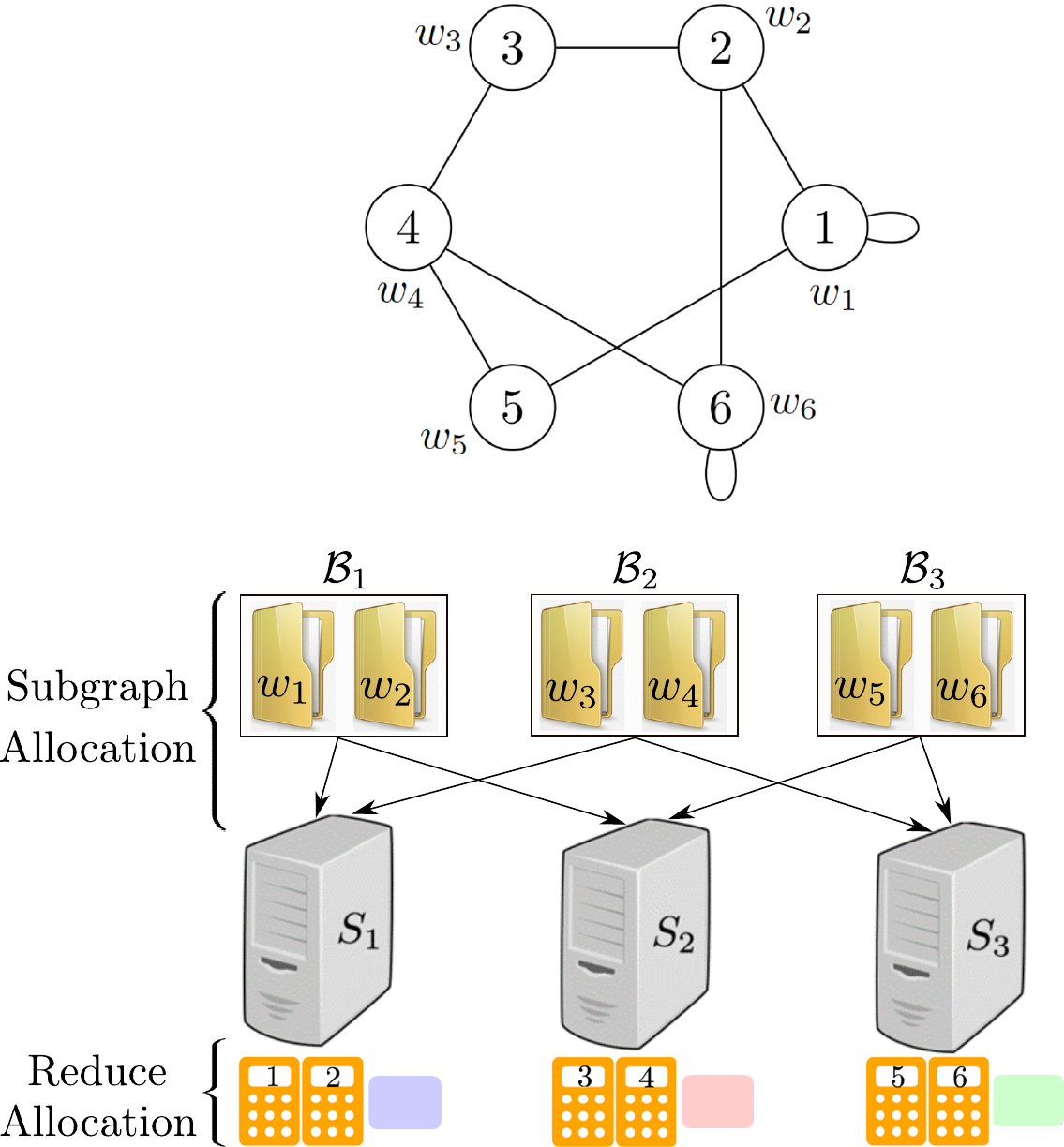} \caption{Illustrating the subgraph allocation and Reduce allocation $A_{\textsf{C}}$ for the example graph with $6$ vertices. The $6$ files are partitioned into $3$ batches and each batch is assigned to a unique subset of $2$ servers. The Reduce functions are partitioned into $3$ sets, one set is assigned to each server.}\label{fig:eg_scheme_a}
    \end{subfigure}\\
    \vspace{.4cm}
    ~ 
    \begin{subfigure}[t]{\linewidth}
        \centering \includegraphics[width=.9\linewidth]{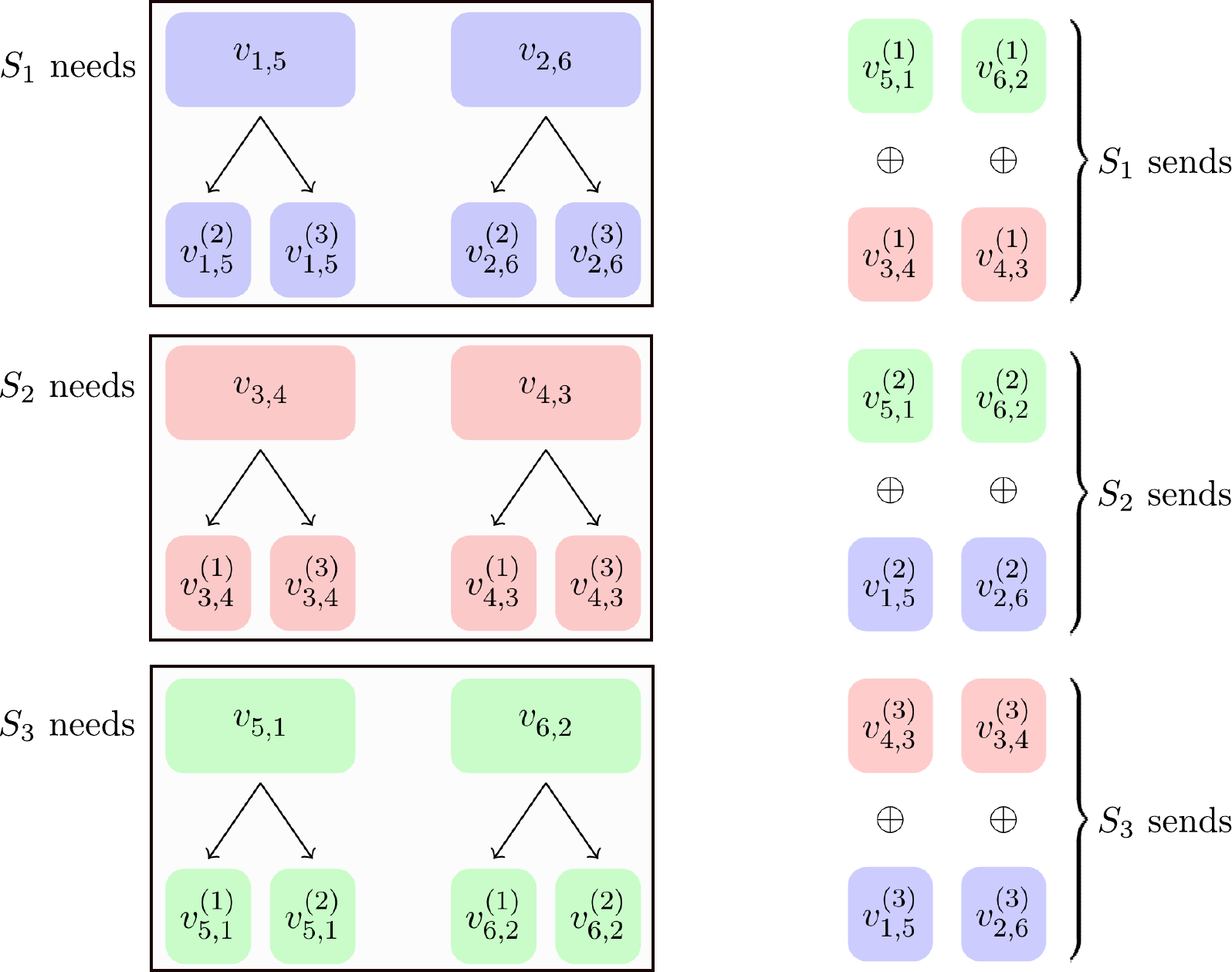}
        \caption{For the subgraph and Reduce allocations $A_{\textsf{C}}$ in Fig. \ref{fig:eg_scheme}(\subref{fig:eg_scheme_a}), we illustrate our proposed coded Shuffle scheme. For each intermediate value needed by a server, each of the remaining two servers is assigned the task of communicating a segment which is one-half of the intermediate value. The servers create a table of the segments that they are assigned to send, with each row corresponding to the intermediate values required exclusively by one of the remaining servers. Each server sends two coded messages, each of which is simultaneously useful for both the remaining servers.}\label{fig:eg_scheme_b}
    \end{subfigure}
     ~     
    \vspace{10pt}
    \caption{Illustration of our proposed scheme.}\label{fig:eg_scheme}
\end{figure}

\noindent\textbf{Coded Shuffle:} As illustrated in Fig. \ref{fig:ex}(\subref{fig:tabmapredex}), the key idea in coded Shuffling is to create coded combinations of locally available intermediate values so that the same message can be useful for many servers simultaneously. Due to the subgraph and Reduce allocation $A_{\textsf{C}}$ described above, every set $\mathcal{F}_j$ of $r$ servers has a unique batch of files  $\mathcal{B}_j$. Thus, all the intermediate values corresponding to the Map computations associated with the files in $\mathcal{B}_j$ are available at every server in $\mathcal{F}_j$ after the Map phase. With this observation, consider without loss of generality the set of servers  $\mathcal{S}=\{1,2,\ldots,r+1\}$. For each server $k\in \mathcal{S}$, let $\mathcal{Z}_{\mathcal{S}\setminus \{k\}}^{k}$ be the set of all intermediate values that are needed by Reduce functions in $k$, and are available exclusively at each server $k'\in \mathcal{S}\setminus \{k\}$, i.e.
\begin{equation}
\label{eq:zset}
\mathcal{Z}_{\mathcal{S}\setminus \{k\}}^{k} = \{ v_{i,j}: (i,j) \in \mathcal{E}, i \in  \mathcal{R}_k, j \in \cap_{k'\in \mathcal{S}\setminus \{k\}}\mathcal{M}_{k'} \}.
\end{equation}
We observe that after the Map phase, server $r+1$ has  $\mathcal{Z}_{\mathcal{S}\setminus \{k\}}^{k}$ for $k\in\{1,\ldots,r\}$. Furthermore, server $1$ has  $\mathcal{Z}_{\mathcal{S}\setminus \{k\}}^{k}$ for $k\in\{2,\ldots,r\}$, server $2$ has  $\mathcal{Z}_{\mathcal{S}\setminus \{k\}}^{k}$ for $k\in\{1,3,\ldots,r\}$, and so on. Therefore, server $r+1$ can create a coded message by selecting one intermediate value each from $\mathcal{Z}_{\mathcal{S}\setminus \{k\}}^{k}$ for $k\in\{1,\ldots,r\}$, and taking a XOR of them. The coded message is simultaneously useful for the servers $\{1,\ldots,r\}$ as each of them can XOR out \textit{its} own missing intermediate value as it has the remaining intermediate values associated with the coded message. Similar arguments hold for the coded messages from other servers within $\mathcal{S}$. 

In light of the above arguments, for each $k\in \mathcal{S}$, each intermediate value $v_{i,j} \in \mathcal{Z}_{\mathcal{S}\setminus \{k\}}^{k}$ is evenly split into $r$ segments $v^{(1)}_{i,j},\cdots,v^{(r)}_{i,j}$, each of size $\frac{T}{r}$ bits. Each segment is associated with a distinct server in $\mathcal{S}\setminus\{k\}$, where the segment assignment is based on the order of the indices of the $r$ servers $\mathcal{S}\setminus\{k\}$. Therefore, $\mathcal{Z}_{\mathcal{S}\setminus \{k\}}^{k}$ is evenly partitioned to $r$ sets, which are denoted by  $ \mathcal{Z}_{\mathcal{S}\setminus \{k\},s}^{k}$ for $s \in \mathcal{S}\setminus\{k\}$. Depending on the connectivity of $G$, the number of intermediate values in $\mathcal{Z}_{\mathcal{S}\setminus \{k\}}^{k}$ shall vary, and the maximum possible size of $\mathcal{Z}_{\mathcal{S}\setminus \{k\}}^{k}$ is $\tilde{g}=g\frac{n}{K}=\frac{n^2}{K{{K}\choose{r}}}$. Each server $s \in \mathcal{S}$ creates an $r \times \tilde{g}$ table and fills that out with segments which are associated with it. Each row of the table is \textit{filled from left} by the segments in one of the sets $\mathcal{Z}_{\mathcal{S}\setminus \{k\},s}^{k},$ where $k\in \mathcal{S}\setminus \{s\}$ (see Fig. \ref{fig:table}). Then, server $s$ broadcasts the XOR of all the segments in each non-empty column of the table, where for each non-empty column, \textit{the empty entries are zero padded}. Clearly, there exist at most $\tilde{g}$ of such coded messages. The process is carried out similarly for every other subset $\mathcal{S}\subseteq [K]$ of servers with $|\mathcal{S}|=r+1$. 

After the Shuffle phase, for each multicast group of $r+1$ servers, all but one intermediate values contributed in each coded message are locally available. Moreover, all possible subsets of multicast servers have sent their corresponding messages. Therefore, each server can recover all of the intermediate values associated with its assigned set of Reduce functions using the received coded messages and the locally computed intermediate values. Thus, our proposed coded Shuffling scheme is feasible, i.e. for any given graph, and subgraph and Reduce allocation $A_{\textsf{C}}$, our proposed Shuffling enables each server to compute all the Reduce functions assigned to it. 

\begin{remark}
The proposed scheme carefully aligns and combines the existing intermediate values to benefit from the coding opportunities. This resolves the issue posed by the asymmetry in the data requirements  of the Reducers which is one of the main challenges in moving from the general MapReduce framework in \cite{li2017fundamental} to graph analytics.
\end{remark}

In Fig. \ref{fig:eg_scheme}(\subref{fig:eg_scheme_b}), every intermediate value in $\mathcal{Z}_{\{1,2\}}^{3} = \{ v_{5,1},v_{6,2} \}$ is split into $r=2$ segments, each associated with a distinct server in $\{1,2\}$. This is done similarly for servers 1 and 2. Then, servers 1, 2, and 3 broadcast their coded messages   $X_1=\{v^{(1)}_{5,1}\oplus v^{(1)}_{4,3}, v^{(1)}_{3,4}\oplus v^{(1)}_{6,2}\}$, $X_2=\{v^{(2)}_{5,1}\oplus v^{(1)}_{1,5},v^{(2)}_{6,2}\oplus v^{(1)}_{2,6}\}$, and $X_3=\{v^{(2)}_{4,3}\oplus v^{(2)}_{1,5},v^{(2)}_{3,4}\oplus v^{(2)}_{2,6}\}$, respectively. All three servers can recover their missing intermediate values. For instance, server 3 needs $v_{5,1}$ to carry out the Reduce function associated with vertex $5$. Since it has already Mapped vertices $3$ and $5$, intermediate values $v_{4,3}$ and $v_{1,5}$ are available locally. Server 3 can recover $v^{(1)}_{5,1}$ and $v^{(2)}_{5,1}$ from $v^{(1)}_{5,1}\oplus v^{(1)}_{4,3}$ and $v^{(2)}_{5,1}\oplus v^{(1)}_{1,5}$, respectively. As each server sends $2$ coded messages to other servers and each coded message is half the size of an intermediate value, therefore, the overall normalized communication load is $\frac{3}{36}$, which is two times better than the normalized communication load for uncoded Shuffling.

\subsection{Proof of Achievability of Theorem \ref{thm:ERCDC}}\label{subsec:ER}
We now analyze the performance of our proposed coded scheme in Section \ref{sec:GCDC} for the Erdös-Rényi random graph model to prove the achievability of Theorem \ref{thm:ERCDC}. For our proposed subgraph and Reduce allocation $A_{\textsf{C}}$, we first compute the average communication for uncoded Shuffle where no coding is utilized during the Shuffle phase. 

\noindent{Uncoded Shuffle:} Given the subgraph and Reduce allocation $A_{\textsf{C}}$, consider a server $k\in [K]$. Due to symmetry, the total expected communication load is sum of the communication loads of each server. Hence we can focus on finding the communication load of server $1$. Note that there are $n/K$ Reducers assigned to server $1$, and $\frac{r n}{K}$ Mappers assigned to server $1$. Therefore, for each Reducer in server $1$, the expected communication required is $(p n-p\frac{r n}{K})T$. Summing over the expected communication loads for all the Reducers in server $1$ and appropriate normalization, the total expected communication load for server $1$ is $\frac{n}{K}(p n-p\frac{r n}{K})T$. Summing over all the $K$ servers, we get the average normalized communication load for the uncoded Shuffle as follows:
\begin{align}
\bar{L}^\textsf{UC}_{A_{\textsf{C}}}
&\coloneqq
\mathbb{E}_{\mathcal{G}}[L^\textsf{UC}_{A_{\textsf{C}}}(r,\mathcal{G})]\nonumber\\
&=
K \frac{n}{K} \left(pn-p\frac{rn}{K} \right)T\frac{1}{n^2 T} \nonumber\\
&=
p \left(1-\frac{r}{K} \right), \nonumber
\end{align}
where $L^\textsf{UC}_{A_{\textsf{C}}}(r,G)$ denotes the normalized communication load for uncoded Shuffle for the graph realization $G$ of the Erdös-Rényi random graph model $\mathcal{G}$. 

We now apply our proposed coded Shuffle scheme and compute the induced average communication load. Without loss of generality, we analyze our algorithm by a generic argument for servers $\mathcal{S}=\{1,\cdots,{r+1}\}$ which can be similarly applied for any other set of servers $\mathcal{S}$ with $|\mathcal{S}|=r+1$, due to the symmetric structure induced by the graph model and subgraph allocation and Reduce allocation $A_{\textsf{C}}$. Denote $r+1$ servers as $s_1,\cdots,s_{r+1}$, and consider the messages that $s_1$ is assigned to send within the multicast group $\mathcal{S}$, the coded messages that are sent by other servers within $\mathcal{S}$ are also created similarly. As described in Section \ref{sec:GCDC} and illustrated in Fig. \ref{fig:table}, server $s_{1}$ creates a table of intermediate value segments for transmission. In this table, each row is filled from the left, and for $i\in [r]$, $i$'th row contains the allocated segments for the intermediate values in the set $\mathcal{Z}_{\mathcal{S}\setminus\{s_{i+1}\}, s_{1}}^{s_{i+1}}$. The number of segments in $\mathcal{Z}_{\mathcal{S}\setminus\{s_{i+1}\}, s_{1}}^{s_{i+1}}$, denoted by $\tilde{g}_i$, depends on the connectivity of the graph $G$ and is upper bounded by $\tilde{g}$, the total number of intermediate values in $\mathcal{Z}_{\mathcal{S}\setminus\{s_{i+1}\}}^{s_{i+1}}$ for a completely connected graph. Server $s_{1}$ broadcasts at most $\tilde{g}_{\text{max}}=\max(\tilde{g}_1,\tilde{g}_2,\ldots,\tilde{g}_r)$ coded messages $X^1,\cdots,X^{\tilde{g}_{\text{max}}}$, zero padding the empty entries in the non-empty columns. These coded messages are simultaneously and exclusively useful for the servers $s_2,\cdots,s_{r+1}$. For each non-empty column $j \in [\tilde{g}_{\text{max}}]$, $X^j$ is XOR of at most $r$ non-zero segments of size $\frac{T}{r}$ bits, associated with server $s_{1}$. More formally, for each non-empty column $j \in [\tilde{g}_{\text{max}}]$, we have the following:
\begin{equation}
\label{eq:codMessage}
    X^j = \bigoplus_{i=1}^{r} v^{(1)}_{\alpha(i,j)}.
\end{equation}
In (\ref{eq:codMessage}), for $i\in[r]$ and $j\in[\tilde{g}_i]$, we have used $v^{(1)}_{\alpha(i,j)}$ to denote the non-zero segment in the table in $i$'th row and $j$'th column, while for $j\in\{\tilde{g}_i+1,\tilde{g}_i+2,\ldots,\tilde{g}\}$, $v^{(1)}_{\alpha(i,j)}$ denotes the zero padding segment.   

\begin{figure}[b!]
%\vspace{-0.2cm}
\begin{center}
\begin{tikzpicture}[scale=0.9, transform shape]

\node  at (0,7.5) {$X^1$};
\node  at (0,7) {$\verteq$};
\draw (0,6) node[minimum size=0.5cm,circle,fill=red!20] {\tiny{$v^{(1)}_{\alpha(1,1)}$}};
\draw (0,4) node[minimum size=0.5cm,circle,fill=red!20] {\tiny{$v^{(1)}_{\alpha(2,1)}$}};
\draw (0,0) node[minimum size=0.5cm,circle,fill=red!20] {\tiny{$v^{(1)}_{\alpha(r,1)}$}};
\node  at (0,5) {$\oplus$};
\node  at (0,3) {$\oplus$};
\node  at (0,1) {$\oplus$};
\node  at (0,2) {\vdots};

\node  at (1.5,7.5) {$X^2$};
\node  at (1.5,7) {$\verteq$};
\draw (1.5,6) node[minimum size=0.5cm,circle,fill=red!20] {\tiny{$v^{(1)}_{\alpha(1,2)}$}};
\draw (1.5,4) node[minimum size=0.5cm,circle,draw,dashed] {\tiny{$v^{(1)}_{\alpha(2,2)}$}};
\draw (1.5,0) node[minimum size=0.5cm,circle,fill=red!20] {\tiny{$v^{(1)}_{\alpha(r,2)}$}};
\node  at (1.5,5) {$\oplus$};
\node  at (1.5,3) {$\oplus$};
\node  at (1.5,1) {$\oplus$};
\node  at (1.5,2) {\vdots};

\node  at (3,7.5) {$X^3$};
\node  at (3,7) {$\verteq$};
\draw (3,6) node[minimum size=0.5cm,circle,fill=red!20] {\tiny{$v^{(1)}_{\alpha(1,3)}$}};
\draw (3,4) node[minimum size=0.5cm,circle,draw,dashed] {\tiny{$v^{(1)}_{\alpha(2,3)}$}};
\draw (3,0) node[minimum size=0.5cm,circle,draw,dashed] {\tiny{$v^{(1)}_{\alpha(r,3)}$}};
\node  at (3,5) {$\oplus$};
\node  at (3,3) {$\oplus$};
\node  at (3,1) {$\oplus$};
\node  at (3,2) {\vdots};

\node  at (6,7.5) {$X^{\tilde{g}}$};
\node  at (6,7) {$\verteq$};
\draw (6,6) node[minimum size=0.5cm,circle,draw,dashed] {\tiny{$v^{(1)}_{\alpha(1,\tilde{g})}$}};
\draw (6,4) node[minimum size=0.5cm,circle,draw,dashed] {\tiny{$v^{(1)}_{\alpha(2,\tilde{g})}$}};
\draw (6,0) node[minimum size=0.5cm,circle,draw,dashed] {\tiny{$v^{(1)}_{\alpha(r,\tilde{g})}$}};
\node  at (6,5) {$\oplus$};
\node  at (6,3) {$\oplus$};
\node  at (6,1) {$\oplus$};
\node  at (6,2) {\vdots};

\node  at (4.5,6) {\ldots};
\node  at (4.5,4) {\ldots};
\node  at (4.5,0) {\ldots};

\node  at (-1.5,6) {$P_1:$};
\node  at (-1.5,4) {$P_2:$};
\node  at (-1.5,0) {$P_r:$};

\end{tikzpicture}
\end{center}
\caption{Creating coded messages by aligning the associated intermediate value segments.  }\label{fig:table}
\vspace{.1cm}
\end{figure}

Let $\operatorname{Bern}(p)$ random variable $E_{\alpha(i,j)}$ indicate the existence of the edge $\alpha(i,j)\in \mathcal{V}\times \mathcal{V}$, i.e. $E_{\alpha(i,j)}=1$, if $\alpha(i,j)\in \mathcal{E}$, and $E_{\alpha(i,j)}=0$, otherwise. Clearly, for all vertices $i,j,t,u\in \mathcal{V}$, $E_{\alpha(i,j)}$ is independent of $E_{\alpha(t,u)}$ if $\alpha(i,j)$ and $\alpha(t,u)$ do not represent the same edge, and $E_{\alpha(i,j)}=E_{ \alpha(t,u)}$, otherwise. 
For $i \in [r]$, the random variable $P_i$ is defined as
\begin{equation}\label{eq:Pjdef}
         P_i=\sum_{j=1}^{\tilde{g}} E_{\alpha(i,j)},
\end{equation}
i.e. each $P_i$ is sum of $\tilde{g}$ possibly dependent $\operatorname{Bern}(p)$ random variables. Note that $P_i$'s are not independent in general. By careful alignment of present intermediate values (Fig. \ref{fig:table}), $s_1$ broadcasts $Q$ coded messages each of size $\frac{T}{r}$ bits, where $Q=\max_{i \in [r]} P_i$. Thus, the total coded communication load sent from server $s_1$ exclusively for servers $s_2,\cdots,s_{r+1}$ is $\frac{T}{r}Q$ bits. By similar arguments for other sets of servers, we can characterize the average normalized coded communication load of the proposed scheme as follows:
\begin{align}\label{eq:12}
     \bar{L}^\textsf{C}_{A_{\textsf{C}}}&\coloneqq \mathbb{E}_{\mathcal{G}}[L^\textsf{C}_{A_{\textsf{C}}}(r,\mathcal{G})] 
     =
     \frac{1}{r n^2} K {K-1 \choose r } \Expc[Q],
\end{align}
where $L^\textsf{C}_{A_{\textsf{C}}}(r,G)$ denotes the normalized communication load for the proposed coded Shuffle for the graph realization $G$ of the Erdös-Rényi random graph model $\mathcal{G}$.

 The following lemma asymptotically upper bounds $\Expc[Q]$ and the proof is provided in Section \ref{subsec:proof}.

\begin{lemma}\label{lemma:ExpQ}
For $\text{ER}(n,p)$ graphs with $p = \omega(\frac{1}{n^2})$, we have 
\begin{equation}\label{eq:14}
    \Expc[Q] \leq p\tilde{g} + o(p\tilde{g}). \nonumber
\end{equation}
\end{lemma}
Putting (\ref{eq:12}) and Lemma \ref{lemma:ExpQ} together, we have
\begin{equation}
\label{eq:asymptotachievb}
    L^*(r) \leq \bar{L}^\textsf{C}_{A_{\textsf{C}}} \leq  \frac{1}{r} p \left(1-\frac{r}{K}\right) + o(p),\nonumber
\end{equation}
hence the achievability claimed in Theorem \ref{thm:ERCDC} is proved. Finally, we note that as explained in the uncoded Shuffle algorithm, the average normalized uncoded communication load of the proposed scheme is
$
    \bar{L}^\textsf{UC}_{A_{\textsf{C}}}=p\left(1-\frac{r}{K}\right),
$
which implies that our scheme achieves an asymptotic gain of $r$.

\begin{remark}
As we next show in the proof of Lemma \ref{lemma:ExpQ}, the regime $p=\omega(1/n^2)$ is essential in order to have $p \tilde{g}=\omega(1)$. As  $\tilde{g}=\frac{n^2}{K{{K}\choose{r}}}=\Theta(n^2)$ is a deterministic function of $n$, the regime $p=\omega(1/n^2)$ is needed to get the achievability and asymptotic optimality of Theorem 1.
\end{remark}

%since $\Expc[L_\textsf{UC}(r)]=p(1-\frac{r}{K})=\Theta(p)$.
%\vspace{-2mm}

\subsection{Proof of Lemma \ref{lemma:ExpQ}}\label{subsec:proof}
Before proving Lemma \ref{lemma:ExpQ}, we first present the following lemma that will be used in our proof.

\begin{lemma}\label{lemma:MGbound}
For random variables $\{P_i\}_{i=1}^{r}$ defined in (\ref{eq:Pjdef}), their moment generating functions for $s'>0$ can be bounded by
\begin{equation}
    \Expc \big[e^{s' P_i}\big] \leq (p e^{2s'}+1-p)^{\tilde{g}/2}. \nonumber
\end{equation}
\end{lemma}
\begin{proof}
 Consider a generic random variable of the form (\ref{eq:Pjdef})
 \begin{equation}
    P = \sum_{j=1}^{\tilde{g}} E_j, \nonumber
\end{equation}
where $E_j$'s are $\operatorname{Bern}(p)$ and possibly dependent. However, although $E_j$'s may not be all independent, but dependency is restricted to pairs of $E_j$'s. In other words, for all $1 \leq j \leq \tilde{g}$, $E_j$ is either independent of all $E_{[\tilde{g}] \setminus \{j\}}$, or is equal to $E_\ell$ for some $\ell\in[\tilde{g}] \setminus \{j\}$ and independent of all $E_{[\tilde{g}]\setminus \{j,\ell\}}$. By merging dependent pairs, we can write 
\begin{equation}
    P = \sum_{j=1}^{\tilde{g}-J} F_j, \nonumber
\end{equation}
where
\begin{enumerate}[(i)]
    \item $F_j$'s are independent,
    \item $\tilde{g}-2J$ of $F_j$'s are $\operatorname{Bern}(p)$,
    \item $J$ of $F_j$'s are $2 \times\operatorname{Bern}(p)$,
  \end{enumerate}
  for some integer $0 \leq J \leq \floor{\frac{\tilde{g}}{2}}$. Now, we can bound the moment generating function of $P$. For $s'>0$,
  \begin{align}
      \Expc\big[e^{s'P}\big] &= \Expc\Big[e^{s'\sum_{j=1}^{J} F_j}\Big]\nonumber\\
                    &= \prod_{j=1}^{\tilde{g}-J} \Expc\big[e^{s'F_j}\big]\nonumber\\
                    &= \big(p e^{s'}+1-p \big)^{\tilde{g}-2J} \big(p e^{2s'}+1-p \big)^{J}\nonumber\\
                    &= \Big[\big(p e^{s'}+1-p \big)^2\Big]^{\tilde{g}/2-J} \big(p e^{2s'}+1-p \big)^{J}\nonumber\\
                    &\overset{(a)}{\leq} \big(p e^{2s'}+1-p \big)^{\tilde{g}/2-J} \big(p e^{2s'}+1-p \big)^{J}\nonumber\\
                    &= \big(p e^{2s'}+1-p \big)^{\tilde{g}/2}, \nonumber
  \end{align}
  where inequality $(a)$ is obtained using Lemma \ref{lemma:ineq} (proof available in Appendix \ref{sub:applemma}).
\end{proof}

We now complete the proof of Lemma \ref{lemma:ExpQ}. For any $s'>0$, we can write 
\begin{align}
    e^{s'\Expc[Q]} &\leq \Expc\big[e^{s'Q}\big]\nonumber\\
    &= \Expc\Big[ \max_{i=1,\cdots,r} e^{s'P_i}\Big]\nonumber\\
    &\leq \Expc \Big[ \sum_{i=1}^{r} e^{s'P_i}  \Big]\nonumber\\
    &= \sum_{i=1}^{r} \Expc\big[e^{s'P_i}\big]\nonumber\\
    & \leq r(p e^{2s'}+1-p)^{\tilde{g}/2}, \nonumber
\end{align}
where the last inequality follows from Lemma \ref{lemma:MGbound}. Taking logarithm from both sides yields 
\begin{align}\label{eq:33}
    \Expc[Q] \leq \frac{1}{s'}\log(r) + \frac{\tilde{g}}{2s'}\log(p e^{2s'}+1-p).
\end{align}
Let us substitute $s=2s'$ in (\ref{eq:33}). Then, 
\begin{align}\label{eq:18}
    \Expc[Q] \leq \frac{1}{s}\log(r^2) + \frac{\tilde{g}}{s}\log(p e^{s}+1-p),
\end{align}
for any $s>0$. Let $\bar{p}=1-p$ and pick
\begin{equation}
    s_* = 2 \sqrt{\frac{\log (r)}{\tilde{g} p\bar{p}}}. \nonumber
\end{equation}
We proceed with evaluation of the right hand side (RHS) of (\ref{eq:18}) at $s=s_*$. We first recall the following Taylor series
\begin{eqnarray}
  \log(1+x)=x-\frac{x^2}{2}+\frac{x^3}{3}-\cdots, &  \text{for } x\in (-1,1],\nonumber\\
  e^x = 1+x+\frac{x^2}{2}+\frac{x^3}{3!}+\cdots,  &  \text{for } x\in \mathbb{R}.\nonumber
\end{eqnarray}
Let $x=p(e^{s_*}-1)$. It is easy to check that for $p = \omega (\frac{1}{n^2})$, we have $x \to 0$ and $s_* \to 0$ as $n \to \infty$. Therefore, for $n \to \infty$ we can write
\begin{align}
    &\quad \log (p e^{s_*}+1-p) \nonumber\\
    &= \log(x+1) \nonumber\\
    &= x-\frac{x^2}{2}+\frac{x^3}{3}-\cdots\nonumber\\
    &= p(e^{s_*}-1) -\frac{p^2(e^{s_*}-1)^2}{2}+\frac{p^3(e^{s_*}-1)^3}{3}-\cdots\nonumber\\
    &= p\big({s_*}+\frac{{s_*}^2}{2}+\frac{{s_*}^3}{3!}+\cdots\big) -\frac{p^2}{2}\big({s_*}+\frac{{s_*}^2}{2}+\frac{{s_*}^3}{3!}+\cdots\big)^2 \nonumber\\
    &\quad +\frac{p^3}{3}\big({s_*}+\frac{{s_*}^2}{2}+\frac{{s_*}^3}{3!}+\cdots\big)^3-\cdots\nonumber\\
    &=p{s_*}+\frac{p\bar{p}}{2}s^2_*+o(ps^2_*). \nonumber
\end{align}
Putting everything together, we have
\begin{align}
     \Expc[Q] &\leq \frac{1}{s_*}\log(r^2) + \frac{\tilde{g}}{s_*}\log(p e^{s_*}+1-p)\nonumber\\
     &= \frac{1}{s_*}\log(r^2) + \frac{\tilde{g}}{s_*}\big( p{s_*}+\frac{p\bar{p}}{2}s^2_*+o(ps^2_*)\big)\nonumber\\
     &= \frac{1}{s_*}\log(r^2) + \tilde{g} p +\frac{\tilde{g}p\bar{p}}{2}s_*+o(\tilde{g}ps_*)\nonumber\\
     &= \tilde{g}p + 2\sqrt{\tilde{g}p\bar{p}\log(r)}+o\left(\sqrt{\tilde{g}p}\right). \nonumber
\end{align}
Recall that $\tilde{g}=\frac{n^2}{K{{K}\choose{r}}}$ which is a deterministic function of $n$. Therefore, we choose $p = \omega(\frac{1}{n^2})$ to have $\tilde{g}p = \omega(1)$ and thus $\sqrt{\tilde{g}p\bar{p}\log(r)}=\Theta\left(\sqrt{\tilde{g}p}\right)=o\left(\tilde{g}p\right)$. Therefore,
$
    \Expc[Q] \leq p\tilde{g} + o(p\tilde{g}),
$
as $n \to \infty$.

\section{Converse for the Erdös-Rényi Model}
\label{sec:ERconverse}
In this section, we prove the asymptotic optimality of our proposed coded scheme for the Erdös-Rényi model, by leveraging the techniques employed in \cite{li2017fundamental}. More precisely, we complete the proof of Theorem \ref{thm:ERCDC} by deriving the lower bound on the best average communication load for the Erdös-Rényi model, that matches the achievability in (\ref{eq:asymptotachievb}).

Let $\mathcal{G}$ be an $\text{ER}(n,p)$ random graph and consider a subgraph and Reduce allocation $A=(\mathcal{M},\mathcal{R})\in \mathcal{A}(r)$, where $\sum_{k=1}^{K}|\mathcal{M}_k|=r n$ and $|\mathcal{R}_k|=\frac{n}{K}$, for all $k\in [K]$. We denote the number of files that are Mapped at $j$ vertices under Map assignment $\mathcal{M}$, as $a^j_{\mathcal{M}}$, for all $j\in [K]$.  The following lemma holds.
\begin{lemma}
$\Expc_{\mathcal{G}}[L_{A}(r,\mathcal{G})] \geq p  \sum_{j=1}^{K} \frac{a^j_{\mathcal{M}}}{n}\frac{K-j}{Kj}.$
\end{lemma}

\begin{proof}
We let intermediate values $v_{i,j}$ be realizations of random variables $V_{i,j}$, uniformly distributed over $\mathbb{F}_{2^T}$. For a random graph $\mathcal{G}=(\mathcal{V},\mathcal{E})$ and subsets $\mathcal{I,J} \subseteq \mathcal{V}=[n]$, define $V^{\mathcal{G}}_{\mathcal{I,J}} = \{V_{i,j}: (i,j)\in \mathcal{E}, i \in \mathcal{I}, j \in \mathcal{J}\}$ as the set of present intermediate values in graph $\mathcal{G}$ corresponding to Reducers in $\mathcal{I}$ and Mappers in $\mathcal{J}$. For a given allocation $A=(\mathcal{M},\mathcal{R}) \in \mathcal{A}(r)$ and a subset of servers $\mathcal{S} \subseteq [K]$, we define $X_{\mathcal{S}}=\{X_k:k \in \mathcal{S}\}$ and $Y^{\mathcal{G}}_{\mathcal{S}}=(V^{\mathcal{G}}_{\mathcal{R}_{\mathcal{S}},:},V^{\mathcal{G}}_{:,\mathcal{M}_{\mathcal{S}}})$, where ``$:$'' denotes all possible indices (which depend on both  allocation and graph realization). As described in Section \ref{sec:distimp}, each coded message is a function of the present intermediate values Mapped at the corresponding server. Moreover, all the intermediate values required by the Reducers are decodable from the locally available intermediate values and received messages at the corresponding server. That is, $H(X_k|V^{\mathcal{G}}_{:,\mathcal{M}_{k}})=0$ and $H(V^{\mathcal{G}}_{\mathcal{R}_{k},:}|X_{[K]},V^{\mathcal{G}}_{:,\mathcal{M}_{k}})=0$ for all servers $k \in [K]$ and graphs $\mathcal{G}$. We denote the number of vertices that are exclusively Mapped by
$j$ servers in $\mathcal{S}$ as $a^{j,\mathcal{S}}_{\mathcal{M}}$, that is
\begin{equation} 
  a^{j,\mathcal{S}}_{\mathcal{M}} \coloneqq   \sum_{ \mathcal{S}_1 \subseteq \mathcal{S} : |\mathcal{S}_1|=j } | (\cap_{k \in \mathcal{S}_1} \mathcal{M}_k) \setminus (\cup_{k' \notin \mathcal{S}_1} \mathcal{M}_{k'})|. \nonumber
\end{equation}

We prove the following claim by induction.

\begin{claim}\label{claim1}
For any subset $\mathcal{S} \subseteq [K]$, 
\begin{equation}\label{eq:claim}
    \Expc_{\mathcal{G}} \Big[H(X_{\mathcal{S}}|Y^{\mathcal{G}}_{\mathcal{S}^c}) \Big] \geq p T \sum_{j=1}^{|\mathcal{S}|} a^{j,\mathcal{S}}_{\mathcal{M}} \frac{n}{K} \frac{|\mathcal{S}|-j}{j}.
\end{equation}
\end{claim}

\begin{proof}
\begin{enumerate}[(i)]
    \item If $\mathcal{S}=\{k\}$, for any $k\in [K]$ and graph ${\mathcal{G}}$  we have $H(X_{\mathcal{S}}|Y^{\mathcal{G}}_{\mathcal{S}^c}) \geq 0 $. Therefore, 
    \begin{equation}
        \Expc_{\mathcal{G}} \Big[H(X_{\mathcal{S}}|Y^{\mathcal{G}}_{\mathcal{S}^c}) \Big] \geq 0 = p T \sum_{j=1}^{1} a^{1,\mathcal{S}}_{\mathcal{M}} \frac{n}{K} \frac{1-1}{1}. \nonumber
    \end{equation}
    \item Assume that claim (\ref{eq:claim}) holds for all subsets of size $S_0$.  For any subset $\mathcal{S} \subseteq [K]$ of size $S_0+1$, the following steps hold:
    \begin{align}
        &\quad H(X_{\mathcal{S}}|Y^{\mathcal{G}}_{\mathcal{S}^c})\nonumber\\
        &= \frac{1}{|S|}\sum_{k\in S}H(X_{\mathcal{S}},X_k|Y^{\mathcal{G}}_{\mathcal{S}^c})\nonumber\\
        &=\frac{1}{|S|}\sum_{k\in S}(H(X_{\mathcal{S}}|X_k,Y^{\mathcal{G}}_{\mathcal{S}^c})+H(X_k|Y^{\mathcal{G}}_{\mathcal{S}^c}))\label{eq:intermed1}\\
        &\geq\frac{1}{|S|}\sum_{k\in S}H(X_{\mathcal{S}}|X_k,Y^{\mathcal{G}}_{\mathcal{S}^c})+\frac{1}{|S|}H(X_{\mathcal{S}}|Y^{\mathcal{G}}_{\mathcal{S}^c}).\label{eq:intermed}
    \end{align}
    where (\ref{eq:intermed}) follows from (\ref{eq:intermed1}) using chain rule and conditional entropy relations. Simplifying (\ref{eq:intermed}) and using $|S|-1 = S_0$, we have the following: 
    \begin{equation}\label{eq:claim2}
        H(X_{\mathcal{S}}|Y^{\mathcal{G}}_{\mathcal{S}^c}) \geq \frac{1}{S_0} \sum_{k \in \mathcal{S}} H(X_{\mathcal{S}}|V^{\mathcal{G}}_{:,\mathcal{M}_{k}},Y^{\mathcal{G}}_{\mathcal{S}^c}).
    \end{equation}
    Moreover, 
    \begin{align}\label{eq:claim1}
        H(X_{\mathcal{S}}|V^{\mathcal{G}}_{:,\mathcal{M}_{k}},Y^{\mathcal{G}}_{\mathcal{S}^c}) &= H(V^{\mathcal{G}}_{\mathcal{R}_k,:}|V^{\mathcal{G}}_{:,\mathcal{M}_{k}},Y^{\mathcal{G}}_{\mathcal{S}^c})\nonumber\\
        &\quad +H(X_{\mathcal{S}}|V^{\mathcal{G}}_{:,\mathcal{M}_{k}},V^{\mathcal{G}}_{\mathcal{R}_k,:},Y^{\mathcal{G}}_{\mathcal{S}^c}).
    \end{align}
We can lower bound expected value of the first RHS term in (\ref{eq:claim1}) as follows
\begin{align}
    &\quad \Expc_{\mathcal{G}} \Big[H(V^{\mathcal{G}}_{\mathcal{R}_k,:}|V^{\mathcal{G}}_{:,\mathcal{M}_{k}},Y^{\mathcal{G}}_{\mathcal{S}^c})\Big] \nonumber\\
    &= \Expc_{\mathcal{G}} \left[ \sum_{v \in \mathcal{R}_k} H(V^{\mathcal{G}}_{\{v\},:}|V^{\mathcal{G}}_{\{v\},\mathcal{M}_{k}\cup \mathcal{M}_{\mathcal{S}^c}}) \right]\nonumber\\
    &=  \Expc_{\mathcal{G}} \left[\sum_{ v \in \mathcal{R}_k} |\mathcal{N}(v)| - |\mathcal{N}(v) \cap (\mathcal{M}_{k}\cup \mathcal{M}_{\mathcal{S}^c})| \right] \nonumber\\
    &= \frac{n}{K} p T \sum_{j=0}^{S_0} a^{j,\mathcal{S} \setminus \{k\}}_{\mathcal{M}}\nonumber\\
    &\geq  \frac{n}{K} p T \sum_{j=1}^{S_0} a^{j,\mathcal{S} \setminus \{k\}}_{\mathcal{M}} \label{eq:claim3}.
\end{align}
Expected value of the second term in RHS of (\ref{eq:claim1}) can be lower bounded from the induction assumption:
\begin{align}
    &\quad \Expc_{\mathcal{G}} \Big[H(X_{\mathcal{S}}|V^{\mathcal{G}}_{:,\mathcal{M}_{k}}, V^{\mathcal{G}}_{\mathcal{R}_k,:},Y^{\mathcal{G}}_{\mathcal{S}^c}) \Big] \nonumber\\
    &= \Expc_{\mathcal{G}} \Big[H(X_{\mathcal{S} \setminus \{k\}} | Y^{\mathcal{G}}_{\mathcal{S} \setminus \{k\}}) \Big] \nonumber\\
    &\geq p T \sum_{j=1}^{S_0} a^{j,\mathcal{S} \setminus \{k\}}_{\mathcal{M}} \frac{n}{K} \frac{S_0-j}{j} \label{eq:claim4}.
\end{align}
Putting (\ref{eq:claim2}), (\ref{eq:claim1}), (\ref{eq:claim3}), and (\ref{eq:claim4}) together, we have 
\begin{align}
    &\quad \Expc_{\mathcal{G}} \Big[ H(X_{\mathcal{S}}|Y^{\mathcal{G}}_{\mathcal{S}^c}) \Big] \nonumber\\
    &\geq \frac{1}{S_0} \sum_{k \in \mathcal{S}} \Expc_{\mathcal{G}} \Big[H(X_{\mathcal{S}}|V^{\mathcal{G}}_{:,\mathcal{M}_{k}},Y^{\mathcal{G}}_{\mathcal{S}^c}) \Big] \nonumber\\
    &= \frac{1}{S_0} \sum_{k \in \mathcal{S}} \Expc_{\mathcal{G}} \Big[H(V^{\mathcal{G}}_{\mathcal{R}_k,:}|V^{\mathcal{G}}_{:,\mathcal{M}_{k}},Y^{\mathcal{G}}_{\mathcal{S}^c}) \Big] \nonumber\\
    &\quad +
    \Expc_{\mathcal{G}} \Big[H(X_{\mathcal{S}}|V^{\mathcal{G}}_{:,\mathcal{M}_{k}},V^{\mathcal{G}}_{\mathcal{R}_k,:},Y^{\mathcal{G}}_{\mathcal{S}^c}) \Big]\nonumber\\
    &\geq \frac{1}{S_0} \sum_{k \in \mathcal{S}} \Big(\frac{n}{K} p T \sum_{i=1}^{S_0} a^{i,\mathcal{S} \setminus \{k\}}_{\mathcal{M}} \nonumber\\
    &\quad \quad  \quad \quad \quad \quad +  p T \sum_{j=1}^{S_0} a^{j,\mathcal{S} \setminus \{k\}}_{\mathcal{M}} \frac{n}{K} \frac{S_0-j}{j} \Big) \nonumber\\
    &= p T \sum_{j=1}^{S_0} \frac{n}{K} \frac{1}{j} \sum_{k \in \mathcal{S}} a^{j,\mathcal{S} \setminus \{k\}}_{\mathcal{M}} \nonumber\\
    &= p T \sum_{j=1}^{S_0+1} a^{j,\mathcal{S}}_{\mathcal{M}} \frac{n}{K} \frac{S_0+1-j}{j}. \nonumber
\end{align}

\item Therefore, for any subset $\mathcal{S} \subseteq [K]$, claim (\ref{eq:claim}) holds.
\end{enumerate}
\end{proof}

Now, pick $\mathcal{S}=[K]$. Then, 
\begin{equation}
    \Expc_{\mathcal{G}} \big[L_{A}(r,\mathcal{G}) \big]  \geq \frac{\Expc_{\mathcal{G}} \Big[H(X_{\mathcal{S}}|Y^{\mathcal{G}}_{\mathcal{S}^c}) \Big]}{n^2T}
     \geq p  \sum_{j=1}^{K} \frac{a^j_{\mathcal{M}}}{n}\frac{K-j}{Kj}. \nonumber
\end{equation}
\end{proof}

\iffalse
Let $L^*(r)$ denote the \textit{minimum of average load}, i.e.
\begin{equation}
   L^*(r) = \inf_{\mathcal{A}:|\mathcal{M}_1|+\cdots+|\mathcal{M}_K|=rN} \Expc_{\mathcal{G}}[L^*_{\mathcal{A}}(r,\mathcal{G})]. \nonumber   
\end{equation}

\begin{theorem}
\begin{equation}
   L^*(r) \geq \frac{1}{r}p(1-\frac{r}{K}). \nonumber
\end{equation}
\end{theorem}
\fi
\noindent \textit{Proof of Converse for Theorem \ref{thm:ERCDC}}. First, we use the result in Claim \ref{claim1} and bound the best average normalized communication load as follows:
\begin{align}
   L^*(r) &\geq \inf_{A} \Expc_{\mathcal{G}} \big[L_{{A}}(r,\mathcal{G}) \big]\nonumber\\
   & \geq \inf_{A} p  \sum_{j=1}^{K} \frac{a^j_{\mathcal{M}}}{n}\frac{K-j}{Kj},\nonumber
\end{align}
where the infimum is over all subgraph and Reduce allocations $A=(\mathcal{M},\mathcal{R})\in \mathcal{A}(r)$ for which $\sum_{k=1}^{K}|\mathcal{M}_k|=r n$ and $|\mathcal{R}_k|=\frac{n}{K}$, $\forall k\in [K]$. Additionally, for any Map allocation with computation load $r$, we have the following equations:
\begin{align}
\label{eq:a1}
  \sum_{j=1}^K a^j_{\mathcal{M}} =n, \,\,
  \sum_{j=1}^K j a^j_{\mathcal{M}} =r n.
\end{align}
Using convexity of $\frac{K-j}{Kj}$ in $j$ and (\ref{eq:a1}), the converse is proved as follows:
\begin{align}
 L^*(r) 
 & \geq
 \inf_{A} p  \sum_{j=1}^{K} \frac{a^j_{\mathcal{M}}}{n}\frac{K-j}{Kj} \nonumber\\
 & \geq
 \inf_{A} p \frac{K- \sum_{j=1}^{K} j \frac{a^j_{\mathcal{M}}}{n} }{K\sum_{j=1}^{K} j \frac{a^j_{\mathcal{M}}}{n}} \nonumber\\
 & =
 \frac{1}{r} p \left(1-\frac{r}{K}\right). \nonumber
\end{align}

\section{Achievability for the Power Law Model}
\label{subsec:pl}

We consider a general model for random graphs where the expected degree sequence $\mathbf{d}=(d_1,\cdots,d_n)$ is independently drawn from a power law distribution with exponent $\gamma$, i.e. $\Pr[d_i=d]=cd^{-\gamma}$ for $i\in [n]$ and $d\geq 1$ and proper constant $c$ \cite{chung2004average}. Given the realization of the expected degrees $\mathbf{d}$, for $\rho=\frac{1}{\sum_{i=1}^n d_i}$ and all $i,j\in [n]$, vertices $i$ and $j$ are connected with probability $p_{i,j}=\Prob[(i,j)\in \mathcal{E}]=\rho d_i d_j$,  independently of other edges. We now proceed to analyze the coded and uncoded communication loads averaged over the random connections and random degrees induced by the subgraph and Reduce allocation $A_{\textsf{C}}$ proposed in Section \ref{sec:GCDC}.

Consider the allocation $A_{\textsf{C}}=(\mathcal{M},\mathcal{R})$ and a subset of servers $\mathcal{S} \subseteq [K]$ of size $|\mathcal{S}|=r+1$. According to the proposed scheme in Section \ref{sec:GCDC}, for every server $s \in \mathcal{S}$, servers in $\mathcal{S}\setminus \{s\}$ form a table and construct coded messages using the intermediate values in the sets $\mathcal{Z}_{\mathcal{S}\setminus \{k\}}^{k}$ (defined in (\ref{eq:zset})) where $k \in \mathcal{S}\setminus \{s\}$. Therefore, $r+1$ tables are formed each constructing coded messages of size $\max_{k \in \mathcal{S}\setminus \{s\}} |\mathcal{Z}_{\mathcal{S}\setminus \{k\}}^{k}|\frac{T}{r}$ bits. The total coded load induced by the subset $\mathcal{S}$ (and exclusively for the use of servers in $\mathcal{S}$)  denoted by $L_{A_{\textsf{C}}}^{\textsf{C}}(\mathcal{S})$ is
 \begin{equation}
      L_{A_{\textsf{C}}}^{\textsf{C}}(\mathcal{S}) = \frac{1}{n^2 r}\sum_{s \in \mathcal{S}} \max_{k \in \mathcal{S}\setminus \{s\}} |\mathcal{Z}_{\mathcal{S}\setminus \{k\}}^{k}|. \nonumber
  \end{equation}
However, in uncoded scenarios, denoted by $L_{A_{\textsf{C}}}^{\textsf{UC}}(\mathcal{S})$ the total uncoded load induced by subset $\mathcal{S}$ (and exclusively for the use of servers in $\mathcal{S}$) is
 \begin{equation}
      L_{A_{\textsf{C}}}^{\textsf{UC}}(\mathcal{S}) = \frac{1}{n^2 } \sum_{s \in \mathcal{S}}  |\mathcal{Z}_{\mathcal{S}\setminus \{s\}}^{s}|. \nonumber
  \end{equation}
We have 
 \begin{align}
    |\mathcal{Z}_{\mathcal{S}\setminus \{s\}}^{s}| &= \sum_{i \in \mathcal{R}_s} |\mathcal{N}(i) \cap (\cap_{k' \in \mathcal{S}\setminus \{s\}} \mathcal{M}_{k'})| \nonumber\\
    &= \sum_{{\substack{i \in \mathcal{R}_s \\  {m \in \cap_{k' \in \mathcal{S}\setminus \{s\}} \mathcal{M}_{k'} }}   }}  \mathbbm{1}\{(i,m) \in \mathcal{E}\} \label{eq:idntf},
  \end{align}
where the random Bernoulli $ \mathbbm{1}\{(i,m) \in \mathcal{E}\}$ indicates the realization of the edge connecting vertices $i$ and $m$, i.e. $\Expc[\mathbbm{1}\{(i,m) \in \mathcal{E}\}  |\mathbf{d} ]=\rho d_i d_m$.
We note that $|\mathcal{R}_s|=n/K$ and $|\cap_{k' \in \mathcal{S}\setminus \{s\}} \mathcal{M}_{k'}|=n/{K \choose r}$. Therefore, there are $\tilde{g}=\frac{n^2}{K{K \choose r}}$ Bernoulli summands in (\ref{eq:idntf}) in which every two summands are either independent or equal and independent of other summands. More precisely, (\ref{eq:idntf}) can be decomposed to sum of all independent Bernoulli random variables and sum of dependent ones as follows:
 \begin{align}
    |\mathcal{Z}_{\mathcal{S}\setminus \{s\}}^{s}| &= \sum_{{\substack{i \in \mathcal{R}_s \nonumber\\  {m \in \cap_{k' \in \mathcal{S}\setminus \{s\}} \mathcal{M}_{k'} }}   }}  \mathbbm{1}\{(i,m) \in \mathcal{E}\}\\
    &= \sum_{\mathcal{F}_1 \text{ or } \mathcal{F}_2 \text{ or } \mathcal{F}_3}  \mathbbm{1}\{(i,m) \in \mathcal{E}\}  \nonumber\\
    &\quad
    +2\sum_{{\substack{i,m \in \mathcal{R}_s \cap ( \cap_{k' \in \mathcal{S}\setminus \{s\}} \mathcal{M}_{k'})\\ i<m }  }}  \mathbbm{1}\{(i,m) \in \mathcal{E}\} , \label{eq:decomp}
  \end{align}
  where we denote the events
  \begin{align}
      \mathcal{F}_1 
      &\coloneqq
      \{ i \in \mathcal{R}_s \setminus \cap_{k' \in \mathcal{S}\setminus \{s\}}  \mathcal{M}_{k'} \, , \, m \in \cap_{k' \in \mathcal{S}\setminus \{s\}}  \mathcal{M}_{k'} \}, \nonumber\\
      \mathcal{F}_2 
      &\coloneqq
      \{ i \in \mathcal{R}_s \, , \, m \in \cap_{k' \in \mathcal{S}\setminus \{s\}} \mathcal{M}_{k'} \setminus \mathcal{R}_s \}, \nonumber\\
      \mathcal{F}_3 
      &\coloneqq
      \{i=m \in \mathcal{R}_s \cap ( \cap_{k' \in \mathcal{S}\setminus \{s\}} \mathcal{M}_{k'}) \}. \nonumber
  \end{align}
Note that with this decompostion, all the Bernoulli summands in both terms in (\ref{eq:decomp}) are independent. Assume that the first and second terms in (\ref{eq:decomp}) contain $\tilde{g}-2J$ and $J$ summands respectively.  \\
According to Kolmogorov's strong law of large numbers (Proposition \ref{prop:slln} provided at the end of this section) and given that the second condition in the proposition is satisfied for Bernoullis, we have 
 \begin{align}
    \frac{1}{\tilde{g}-2J}
    \sum_{\mathcal{F}_1 \text{ or } \mathcal{F}_2 \text{ or } \mathcal{F}_3} \mathbbm{1}\{(i,m) \in \mathcal{E}\} - \Expc[\rho d_i d_m ] \xrightarrow{\text{a.s.}}  0, \nonumber
  \end{align}
  and 
   \begin{align}
     \frac{1}{J}\sum_{{\substack{i,m \in \mathcal{R}_s \cap ( \cap_{k' \in \mathcal{S}\setminus \{s\}} \mathcal{M}_{k'})\\ i<m }  }}  \mathbbm{1}\{(i,m) \in \mathcal{E}\} -\Expc[\rho d_i d_m ]  \xrightarrow{\text{a.s.}}  0. \nonumber
  \end{align}
Therefore, size of the set $\mathcal{Z}_{\mathcal{S}\setminus \{s\}}^{s}$ converges almost surely, that is
\begin{align}
    &\quad \frac{1}{\tilde{g}} \left( |\mathcal{Z}_{\mathcal{S}\setminus \{s\}}^{s}| - \Expc\big[ |\mathcal{Z}_{\mathcal{S}\setminus \{s\}}^{s}|\big] \right) \nonumber \\
    &=
    \frac{\tilde{g}-2J}{\tilde{g}}\frac{1}{\tilde{g}-2J}\sum_{\mathcal{F}_1 \text{ or } \mathcal{F}_2 \text{ or } \mathcal{F}_3}  \mathbbm{1}\{(i,m) \in \mathcal{E}\} - \Expc[\rho d_i d_m ]\nonumber \\
    &\quad+
    \frac{J}{\tilde{g}}\frac{1}{J}2\sum_{{\substack{i,m \in \mathcal{R}_s \cap ( \cap_{k' \in \mathcal{S}\setminus \{s\}} \mathcal{M}_{k'})\\ i<m }  }}  \mathbbm{1}\{(i,m) \in \mathcal{E}\} - \Expc[\rho d_i d_m ]\nonumber\\
    &\xrightarrow{\text{a.s.}} 0, \nonumber
  \end{align}
where 
 \begin{align}
    \Expc\big[ |\mathcal{Z}_{\mathcal{S}\setminus \{s\}}^{s}|\big]
    &=
     \sum_{{\substack{i \in \mathcal{R}_s \\  {m \in \cap_{k' \in \mathcal{S}\setminus \{s\}} \mathcal{M}_{k'} }}   }} \Expc[\rho d_i d_m ]\nonumber\\
     &=  
      \Expc \big[\rho\, \text{vol}(\mathcal{R}_s) \text{vol}(\cap_{k' \in \mathcal{S}\setminus \{s\}} \mathcal{M}_{k'}) \big], \nonumber
  \end{align}  
  and $\text{vol}(V)=\sum_{v \in V} d_v$ for any subset of vertices $V \subseteq [n]$. Moreover, 
   \begin{align}
    &\quad \lim_{n \to \infty}  \frac{n}{\tilde{g}} \Expc \big[|\mathcal{Z}_{\mathcal{S}\setminus \{s\}}^{s}| \big] 
    \nonumber\\ 
    &=
    \lim_{n \to \infty}  \Expc \left[(\rho n) \frac{1}{n/K} \text{vol}(\mathcal{R}_s) \frac{1}{n/{K \choose r}} \text{vol}(\cap_{k' \in \mathcal{S}\setminus \{s\}} \mathcal{M}_{k'}) \right]. \label{eq:expcz}
  \end{align}
Each of the terms $\text{vol}(\mathcal{R}_s)$, $\text{vol}(\cap_{k' \in \mathcal{S}\setminus \{s\}} \mathcal{M}_{k'})$ and inverse of $\rho$ are summation of i.i.d power law random variables for which the expected value exists for $\gamma>2$ and  $\Expc[d_1]=\frac{\gamma-1}{\gamma-2}$. Therefore, by strong law of large numbers (Proposition \ref{prop:slln}) each term approaches its average almost surely, that is for $\gamma>2$
\begin{equation}
       \frac{1}{n/K} \text{vol}(\mathcal{R}_s)  \xrightarrow{\text{a.s.}} \Expc[d_1] =  \frac{\gamma-1}{\gamma-2}, \nonumber
\end{equation}       
\begin{equation}       
       \frac{1}{n/{K \choose r}} \text{vol}(\cap_{k' \in \mathcal{S}\setminus \{s\}} \mathcal{M}_{k'})  \xrightarrow{\text{a.s.}} \Expc[d_1] =  \frac{\gamma-1}{\gamma-2}. \nonumber
  \end{equation}
\begin{equation}
       \rho n  \xrightarrow{\text{a.s.}} \frac{1}{\Expc[d_1]} =  \frac{\gamma-2}{\gamma-1}. \nonumber
\end{equation}   
Plugging into (\ref{eq:expcz}), we have 
$\lim_{n \to \infty}  \frac{n}{\tilde{g}} \Expc \big[|\mathcal{Z}_{\mathcal{S}\setminus \{s\}}^{s}| \big] = \left(\frac{\gamma-1}{\gamma-2}\right). 
$ Therefore, 
$
\frac{n}{\tilde{g}}|\mathcal{Z}_{\mathcal{S}\setminus \{s\}}^{s}|
\xrightarrow{\text{a.s.}}
\left(\frac{\gamma-1}{\gamma-2}\right)
$ 
for any $s \in \mathcal{S}$ and $\mathcal{S} \subseteq[K]$. Putting all together, we have for $\gamma>2$, 
 \begin{align}
    \lim_{n \to \infty} n\Expc[L_{A_{\textsf{C}}}^{\textsf{UC}}(\mathcal{S})] &= \lim_{n \to \infty} \frac{n}{n^2 } \sum_{s \in \mathcal{S}}  \Expc\big[|\mathcal{Z}_{\mathcal{S}\setminus \{s\}}^{s}|\big] \nonumber\\
    &= \frac{1}{K{K \choose r}} \lim_{n \to \infty} \sum_{s \in \mathcal{S}} \frac{n}{\tilde{g}} \Expc \big[|\mathcal{Z}_{\mathcal{S}\setminus \{s\}}^{s}| \big] \nonumber\\
    &= \frac{r+1}{K{K \choose r}} \left(\frac{\gamma-1}{\gamma-2}\right). \nonumber
  \end{align}
 Therefore, denoted by $L_{A_{\textsf{C}}}^{\textsf{UC}}$ the total uncoded communication load, we have
  \begin{align}
      \lim_{n \to \infty} n\Expc[L_{A_{\textsf{C}}}^{\textsf{UC}}] 
      &= \lim_{n \to \infty} \sum_{{\substack{\mathcal{S} \subseteq[K] \\ n|\mathcal{S}|=r+1}}} \Expc[L_{A_{\textsf{C}}}^{\textsf{UC}}(\mathcal{S})]\nonumber\\
      &= {K \choose r+1} \frac{r+1}{K{K \choose r}}  \left(\frac{\gamma-1}{\gamma-2}\right) \nonumber\\
      &= \big(1-\frac{r}{K}\big) \left(\frac{\gamma-1}{\gamma-2}\right). \nonumber
  \end{align}
For the coded scheme, we have 
\begin{align}
       \lim_{n \to \infty} n\Expc[L_{A_{\textsf{C}}}^{\textsf{C}}(\mathcal{S})] &=  \lim_{n \to \infty}  \frac{n}{n^2 r}\sum_{s \in \mathcal{S}} \Expc \left[\max_{k \in \mathcal{S}\setminus \{s\}} |\mathcal{Z}_{\mathcal{S}\setminus \{k\}}^{k}|\right] \nonumber\\
       & \leq \lim_{n \to \infty}  \frac{n(r+1)}{n^2 r} \Expc\left[\max_{s \in \mathcal{S}} |\mathcal{Z}_{\mathcal{S}\setminus \{s\}}^{s}|\right] \nonumber\\
       &= \frac{r+1}{rK{K \choose r}}  \left(\frac{\gamma-1}{\gamma-2}\right). \label{eq:codedl}
  \end{align}
The last equality follows the fact that $
\frac{n}{\tilde{g}} \max_{s \in \mathcal{S}} |\mathcal{Z}_{\mathcal{S}\setminus \{s\}}^{s}|
\xrightarrow{\text{a.s.}}
\left(\frac{\gamma-1}{\gamma-2}\right),
$
since 
$
\frac{n}{\tilde{g}} |\mathcal{Z}_{\mathcal{S}\setminus \{s\}}^{s}|
$
converges almost surely for any $s \in \mathcal{S}$. Plugging into (\ref{eq:codedl}), the expected coded load is 
\begin{align}
       \lim_{n \to \infty} n\Expc[L_{A_{\textsf{C}}}^{\textsf{C}}] 
       &=
       \lim_{n \to \infty}n \sum_{{\substack{\mathcal{S} \subseteq[K] \\ |\mathcal{S}|=r+1}}} \Expc[L_{A_{\textsf{C}}}^{\textsf{C}}(\mathcal{S})] \nonumber\\
       &\leq {K \choose r+1} \frac{r+1}{rK{K \choose r}}  \left(\frac{\gamma-1}{\gamma-2}\right) \nonumber\\
       &= \frac{1}{r}\left(1-\frac{r}{K}\right) \left(\frac{\gamma-1}{\gamma-2}\right), \nonumber
  \end{align}
which yields
\begin{equation}
    \lim_{n \to \infty}\frac{n L^*(r)}{(\frac{\gamma-1}{\gamma-2})  } \leq \lim_{n \to \infty}\frac{n \Expc[L_{A_{\textsf{C}}}^{\textsf{C}}]}{(\frac{\gamma-1}{\gamma-2})  } \leq \frac{1}{r}\left(1-\frac{r}{K}\right).  \nonumber
\end{equation}
Comparing the coded load with uncoded load proves the achievability of gain $r$ for the power law model.

\begin{proposition}[Kolmogorov's Strong Law of Large Numbers \cite{sen1994largesample,loeve1977probability}]\label{prop:slln}
Let $X_1,X_2,\cdots,X_n,\cdots$ be a sequence of independent random variables with $|\Expc[X_n]| < \infty$ for $n \geq 1$. Then 
\begin{equation}
    \frac{1}{n} \sum_{i=1}^{n} \big( X_i - \Expc[X_i] \big) \xrightarrow{\text{a.s.}} 0,  \nonumber
\end{equation}
if one of the following conditions are satisfied:
\begin{enumerate}
    \item $X_i$'s are identically distributed,
    \item $\forall n$, $var(X_n) < \infty$ and $\sum_{n=1}^{\infty} \frac{var(X_n)}{n^2} < \infty$.
\end{enumerate}
\end{proposition}

\section{Experiments over Amazon EC2 Clusters}
\label{sec:expe}
In this section, we demonstrate the practical impact of our proposed coded scheme via experiments over Amazon EC2 clusters. We first present our implementation choices and experimental scenarios. Then, we discuss the results and provide some remarks. Implementation codes are available at \cite{pagerank2018github}.

\subsection{Implementation Details}
We implement one iteration of the popular PageRank algorithm (Example \ref{ex:PageRank}), for a real-world graph as well as artificially generated graphs. For real-world dataset, we use TheMarker Cafe Dataset \cite{markercafe}. For generating artificial graph datasets, we use the Erdös-Rényi model, where each edge in the graph is present with probability $p$. We consider the following three scenarios:
\begin{itemize}
\item \textbf{Scenario $\mathbf{1}$}: We use a subgraph of size $n=69360$ of TheMarker Cafe Dataset \cite{markercafe}. The computing cluster consists of $K=6$ servers and one master with communication bandwidth of $100$ Mbps at each server.

\item \textbf{Scenario $\mathbf{2}$}: We generate a graph using the Erdös-Rényi model with $n=12600$ vertices and $p=0.3$. The computing cluster consists of $K=10$ servers and one master with communication bandwidth of $100$ Mbps at each server.

\item \textbf{Scenario $\mathbf{3}$}: We generate a graph using the  Erdös-Rényi model with $n=90090$ vertices and $p=0.01$. The computing cluster consists of $K=15$ servers and one master with communication bandwidth of $100$ Mbps at each server.
\end{itemize}

For each scenario, we carry out PageRank implementation for different values of the computation load $r$. The case of $r=1$ corresponds to the conventional PageRank implementation, where each vertex $i\in \mathcal{V}=[n]$ is stored at exactly one server and $\mathcal{M}_k=\mathcal{R}_k$ for each server $k\in [K]$, i.e. the Map and Reduce tasks associated with any vertex $i$ take place in the same server. For $r>1$, we increase the computation load until the overall execution time starts increasing. 

We now describe our implementation choices. We use Python with \texttt{mpi4py} package. In all of our experiments, master is of type r4.large and servers are of type m4.large. For Scenario $2$ and Scenario $3$, we use a sample from the Erdös-Rényi model. This process is carried out using a c4.8xlarge server instance. For each scenario, the graphs are processed and subgraph allocation is done as a pre-processing step. For $r=1$, the graph is partitioned into smaller instances which have equal numbers of vertices. Each such partition consists of two Python \texttt{lists}, one that consists of the vertices that will be Mapped by the corresponding server, and the other one that consists of the neighborhood information of each vertex to be Mapped. The position of the neighborhood \texttt{tuple} in the neighborhood \texttt{list} is same as the position of the corresponding vertex in the vertex \texttt{list}, so that one can iterate over the two together during the Map stage. For $r>1$, the graph is divided into $K \choose r$ batches, where each batch consists of equal numbers of vertices. Then each batch is included in the subgraph of the corresponding set of $r$ servers. This way, we get a a computation load of $r$. 

The overall execution consists of the following phases:
\begin{enumerate}
    \item \textbf{Map}: Without loss of generality, the rank for each vertex is initialized to $\frac{1}{n}$. Each server goes over its subgraph and Maps the rank associated with a vertex to intermediate values that are required by the neighboring vertices during the Reduce stage. Each intermediate value consists of  key-value pair, where the key is an integer storing the vertex ID, while the value is a real number storing the associated value. Based on the vertex ID, the intermediate value is associated with the partition where the vertex is Reduced, which is obtained by hashing the vertex ID. For each partition, a separate $\texttt{list}$ is created for storing keys and values.  
    
\item \textbf{Encode/Pack}: In conventional PageRank, no encoding is done as the transfer of intermediate values is done directly. For $r>1$, coded multicast packets are created using the proposed encoding scheme. Transmission data is serialized before Shuffling.

\item \textbf{Shuffle}: At any time, only one server is allowed to use the network for transmission. In conventional PageRank, each server unicasts its message to different servers, while for $r>1$, the communication takes place in multicast groups. For any multicast group, each server takes its turn to  broadcast its message to all the remaining servers in the group. 

\item \textbf{Unpack/Decode}: The messages received during the Shuffle phase are de-serialized. For $r>1$, each server decodes the coded packets received from other servers in accordance with the proposed coded scheme to recover the intermediate values. After the decoding phase, all intermediate values that are needed for Reduce phase are available at the servers.

\item \textbf{Reduce}: Each server goes over its set of vertices that it needs to Reduce and updates the corresponding PageRank values. In conventional PageRank, for any vertex $i\in \mathcal{V}$, the Map and Reduce operations associated with it are done at the same server. Therefore, no further data transmission is needed to communicate the updated ranks for the Map phase in next iteration. In the proposed coded scheme, message passing is done in order to transmit the updated PageRanks to the Mappers.
\end{enumerate}

Next, we discuss the results of our experiments.

\subsection{Experimental Results}
We now present the results from our experiments. The overall execution times for the three scenarios have been presented in Fig. \ref{fig:exp}.\footnote{The Map time includes the time spent in Encode/Pack stage, while the Unpack stage is combined with Reduce phase.} We make the following observations from the results:
\begin{figure}[htp]
\begin{subfigure}{\linewidth}
\centering
\includegraphics[width=.8\linewidth]{exp_real.eps}
\caption{Scenario $1$}
\vspace{0.4cm}
\label{fig:exp_real}
\end{subfigure} \\
\begin{subfigure}{\linewidth}
\centering
\includegraphics[width=.8\linewidth]{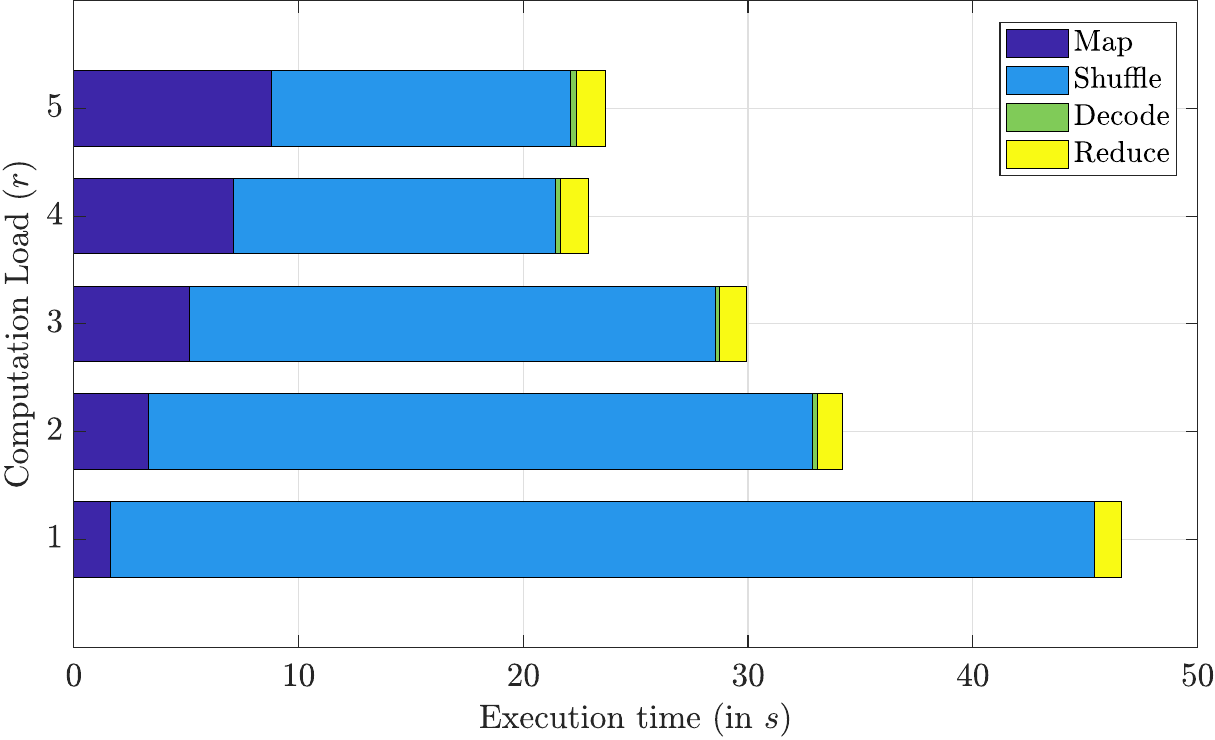}
\caption{Scenario $2$}
\vspace{0.4cm}
  \label{fig:erdos_10}
\end{subfigure}\\[1ex]
\begin{subfigure}{1\linewidth}
\centering
\includegraphics[width=0.8\linewidth]{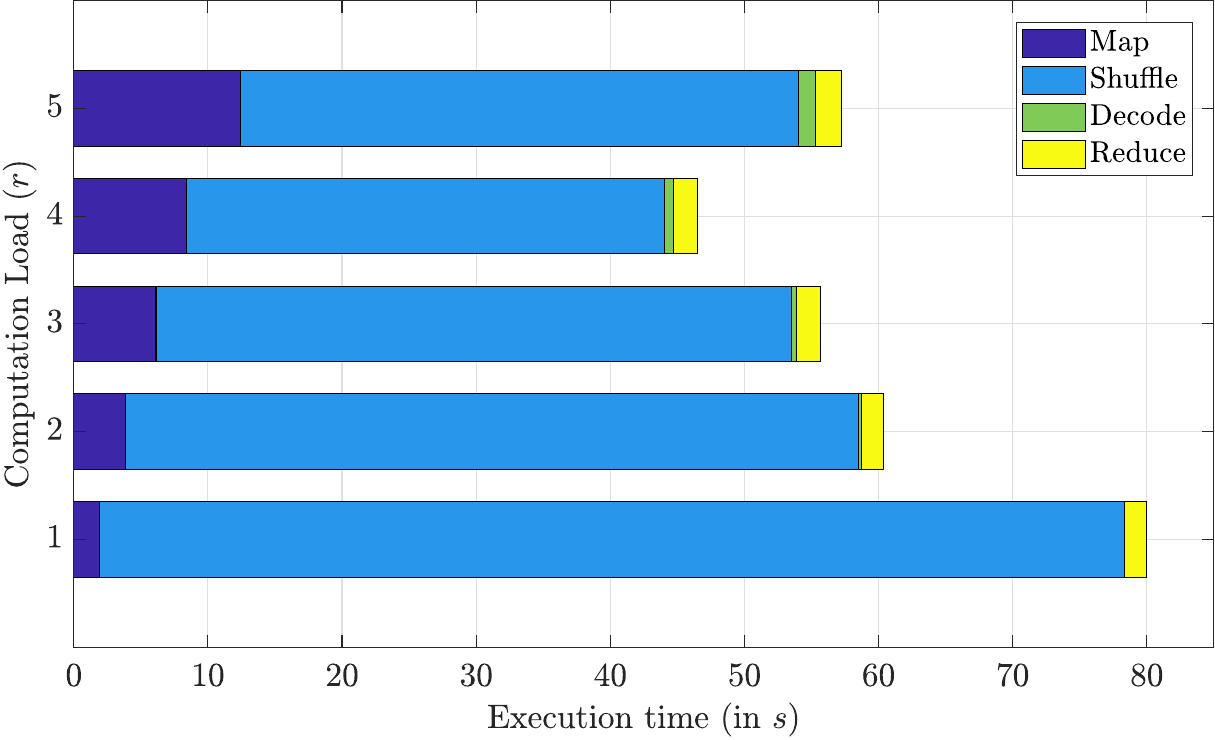}
\caption{Scenario $3$}
\vspace{0.4cm}
  \label{fig:erdos_15}
\end{subfigure}
\caption{Overall execution times for distributed PageRank implementation for different computation load for the three scenarios.}
 \label{fig:exp}
  %\vspace{-.5cm}
\end{figure}
\begin{itemize}
\item As demonstrated in Fig. \ref{fig:exp}(\subref{fig:exp_real}), maximum gain for Scenario $1$ is obtained with a computation load of $r=5$. Our proposed scheme achieves a speedup of $43.4\%$ over conventional PageRank implementation ($r=1$) and a speedup of $25.5\%$ over the single server implementation ($r=6$).
\item For Scenarios $2$ and $3$, the optimal gain is obtained for $r=4$, after which the overall execution time increases due to saturation of gain in Shuffling time and large Map time. As demonstrated by Fig. \ref{fig:exp}(\subref{fig:erdos_10}) and Fig. \ref{fig:exp}(\subref{fig:erdos_15}), our proposed scheme achieves speedups of $50.8\%$ and $41.8\%$ for Scenarios $2$ and $3$ respectively, in comparison to the conventional PageRank. 
\item As demonstrated by Fig. \ref{fig:exp}, Shuffle phase dominates the overall execution time in the naive implementation of PageRank. By increasing the computation load, our proposed coded scheme leverages extra computing in the Map phase to slash the Shuffle phase, thus speeding up the overall execution time.
\item Theoretically, we demonstrated that by increasing the computation load by $r$, we slash the expected communication load in Shuffle phase by nearly $r$. Here, we empirically observe that due to large size of the graph model, we have a similar trade-off between computation load and communication load for each sample of the graph model as well. 
\item While the Map phase increases almost linearly with $r$, the overall gain begins to saturate, since the Shuffle phase does not decrease linearly with $r$. This is because as we increase $r$, the overheads in multicast data transmissions increase and start to dominate the overall Shuffling time. Furthermore, unicasting one packet is smaller than the time for broadcasting the same packet to multiple servers  \cite{lee2016speeding}. 
\end{itemize}
\begin{remark}
The overall execution time can be approximated as follows:
\begin{equation}
\label{eq:opt_r}
T_{\text{Total}}(r) \approx r T_{\text{Map}} +T_{\text{Shuffle}}/r +T_{\text{Reduce}}, 
\end{equation}
where $T_{\text{Map}}$, $T_{\text{Shuffle}}$ and $T_{\text{Reduce}}$ are the Map, Shuffle and Reduce times for the naive MapReduce implementation. For selecting the computation load for coded implementation, one heuristic \cite{li2017fundamental} is to choose $r$ that is the nearest integer to the minimizer $r^*$ of (\ref{eq:opt_r}) where
$$r^*=\sqrt{\frac{T_{\text{Shuffle}}}{T_{\text{Map}}}}=\argmin_{r} T_{\text{Total}}(r).$$ For instance, in Scenario $2$, $T_\text{Map}=1.649$, $T_\text{Shuffle}=43.78$ and $r^*=5.15$. As demonstrated by Fig. \ref{fig:exp}(\subref{fig:erdos_10}), a computation load of $r=5$ gives close to the optimal performance attained at $r=4$.  
\end{remark}

\section{Conclusion and Future Work}
We described a mathematical model for graph based MapReduce computations and demonstrated how coding theoretic strategies can be employed to substantially reduce the communication load in distributed graph analytics. Our results reveal that an inverse-linear trade-off exists between computation load and communication load in distributed graph processing. This trade-off can be used to leverage additional computing resources and capabilities to  alleviate  the  costly communication bottleneck in distributed graph processing systems. 

As a key contribution of this work, we developed a novel coding scheme that systematically injects structured redundancy in the computation phase to enable coded multicasting opportunities during message exchange between servers, reducing the communication load substantially in large-scale graph processing. For theoretical analysis, we considered random graph models, and proved that our proposed scheme enables an asymptotically inverse-linear trade-off between computation load and average normalized communication load for two popular random graph models -- Erdös-Rényi model, and power law model. Furthermore, for the Erdös-Rényi model, we provided proof for a matching converse, showing the optimality of our proposed scheme. We also carried out experiments over Amazon EC2 clusters to corroborate our claims using real-world as well as artificial graphs, demonstrating speedups of up to $50.8\%$ in the overall execution time of PageRank over the conventional approach. Additionally, we specialized our coded scheme and extended our theoretical results to two other random graph models – random bi-partite model, and stochastic block model. Our specialized schemes asymptotically enable inverse-linear trade-offs between computation and communication loads in distributed
graph processing for these popular random graph models as well. We complemented the achievability results with converse bounds
for both of these models.

One of the major differences from prior frameworks such as Pregel is the use of combiners before Shuffling \cite{malewicz2010pregel}, where the intermediate values that are Mapped at any server are combined at the server depending on the target Reducer computations. Our proposed schemes can be applied on top of combiners, and it is an interesting future direction to explore this in detail. The case with fully connected graphs can be solved using the scheme proposed in the recent work of  \cite{li2018compressed}, which shows that the coding gain can be achieved on top of the gain from combiners. For the general MapReduce computation model considered in \cite{li2017fundamental}, the proposed scheme in \cite{li2018compressed} utilizes the techniques of combiners as well as coding across intermediate results, which provides a Shuffling gain which is multiplicative of the gains from combiners and coding. Furthermore, we focused on subgraph allocation and Reduce allocation schemes that are oblivious to graph realizations. Our motivation came from popular graph processing frameworks such as Pregel~\cite{malewicz2010pregel}, where partitioning of graphs is solely based on the vertex ID and not on the vertex neighborhood density. Also, designing subgraph allocation, Reduce allocation and Shuffling schemes for characterizing the minimum communication load in (\ref{eg:compcommGSample}) is NP-hard in general. It might, however, be an interesting future direction to explore the development of coded schemes that allocate resources \textit{after} looking at the graph.

\bibliographystyle{ieeetr}
%vspace{-2mm}
\bibliography{biblio}

\begin{appendices}

\section{Achievability for the Random Bi-partite Model}
\label{subsec:RB}

In this Section, we specialize our proposed scheme in Section \ref{sec:achvblty} for the random bi-partite model and prove the achievability of Theorem \ref{thm:rndmbiprt}. Consider RB$(n_1,n_2,q)$ graph $\mathcal{G}=(\mathcal{V}_1 \cup \mathcal{V}_2,\mathcal{E})$ with $n=n_1+n_2$,  $|\mathcal{V}_1|=n_1=\Theta(n)$, and $|\mathcal{V}_2|=n_2=\Theta(n)$ where $|n_1 - n_2| = o(n)$.  The prior knowledge of the bi-partite structure of the graph implies that Reduction of vertices in $\mathcal{V}_1$ depends only on the Mappers in $\mathcal{V}_2$. Therefore, the two operations would better be assigned to the same set of servers. Inspired by that argument, we describe subgraph and Reduce allocations as follows. We divide the total $K$ servers into two sets of $K_1 = \frac{n_1}{n} K$ and $K_2 = \frac{n_2}{n} K$ servers. Assume $n_1 \geq n_2$.
\begin{enumerate}[(I)]
    \item Mappers in $\mathcal{V}_1$ and Reducers in $\mathcal{V}_2$ are distributedly allocated to $K_1$ servers according to the allocation scheme proposed in Section \ref{sec:GCDC}. Each of the $K_1$ servers Maps  $n_1 \frac{r}{K_1} = n \frac{r}{K}$ vertices (in $\mathcal{V}_1$) and Reduces $\frac{n_2}{K_1} = \frac{n_2}{n_1} \frac{n}{K}$ vertices (in $\mathcal{V}_2$). Note that although each server in $K_1$ is loaded at its capacity with $ n \frac{r}{K}$ Mappers, these servers are assigned $\frac{n_2}{n_1} \frac{n}{K} \leq \frac{n}{K}$ Reducers which implies more Reducers can be assigned to these servers. 
    \item Next we allocate the Mappers in $\mathcal{V}_2$ to the other set of $K_2$ servers similar to Mappers in $\mathcal{V}_1$. According to our pick for $K_2$ and the allocation scheme proposed in Section \ref{sec:GCDC}, each server in $K_2$ is assigned with $n_2 \frac{r}{K_2} = n \frac{r}{K}$ vertices (in $\mathcal{V}_2$). To allocate the $n_1$ Reductions in $\mathcal{V}_1$ to the $K_2$ servers, we note that these servers can accommodate at most $K_2 \frac{n}{K} = n_2$ Reductions which is less than $n_1$. To allocate all Reductions, we use the remaining Reduction space in the $K_1$ servers. More precisely, we first allocate $n_2$ out of the total $n_1$ Reductions in $\mathcal{V}_1$ to the $K_2$ servers.
    \item Finally, we allocate the remaining $n_1 - n_2$ vertices to the $K_1$ servers.
\end{enumerate}
All in all, each of the $K$ servers is now assigned with $nr/K$ Mappers and $n/K$ Reducers. We denote this allocation by $\tA \in \mathcal{A}(r)$. Moreover, coded Shuffling applies the coded scheme proposed in Section \ref{sec:GCDC} for Reducing functions in phases (I) and (II) separately. We also allow uncoded communications for enabling Reductions required in phase (III).

Now, we evaluate the communication load of each of the above phases.  Let $\bar{L}^{\textsf{C1}}_{\tA}$, $\bar{L}^{\textsf{C2}}_{\tA}$ denote the average normalized communication loads for phases (I) and (II); and $\bar{L}^{\textsf{UC3}}_{\tA}$ denote the average normalized communication load regarding phase (III). From the achievability result in Theorem \ref{thm:ERCDC}, for $q = \omega(\frac{1}{n^2})$, we have 
\begin{equation}
\label{eq:bp1}
    \bar{L}^{\textsf{C1}}_{\tA}
    \leq
    \frac{1}{r}q \frac{n_1 n_2}{n^2}  \left(1-\frac{r}{K_1}\right)+o(q), \nonumber
\end{equation}
and
\begin{equation}
\label{eq:bp2}
    \bar{L}^{\textsf{C2}}_{\tA}
    \leq
    \frac{1}{r}q \frac{n^2_2}{n^2}  \left(1-\frac{r}{K_2}\right)+o(q). \nonumber
\end{equation}
As mentioned before, Reduction of the remaining $n_1 - n_2$ vertices in phase (III) is carried out uncoded, which induces the average normalized communication load as follows:
\begin{equation}
\label{eq:bp3}
    \bar{L}^{\textsf{UC3}}_{\tA}
    =
    q \frac{n_2 (n_1 - n_2)}{n^2}. \nonumber
\end{equation}
Putting all together, the proposed achievable scheme has the total average normalized communication load $\bar{L}_{\tA}$ as follows:
\begin{align}
    \bar{L}_{\tA}
    & =
    \bar{L}^{\textsf{C1}}_{\tA}
    +
    \bar{L}^{\textsf{C2}}_{\tA}
    +
    \bar{L}^{\textsf{UC3}}_{\tA}\nonumber\\
    & \leq
    \frac{1}{r}q \frac{n_1 n_2}{n^2}  \left(1-\frac{r}{K_1}\right) \nonumber\\
    &\quad+
    \frac{1}{r}q \frac{n^2_2}{n^2}  \left(1-\frac{r}{K_2}\right)\nonumber\\
    & \quad +
    q \frac{n_2 (n_1 - n_2)}{n^2}
    +o(q). \nonumber
\end{align}
Hence, the achievability claim of Theorem \ref{thm:rndmbiprt} can be concluded as follows:
\begin{align}
\label{eq:achievbRB}
    \limsup_{n \to \infty}\frac{L^*(r)}{q}
    & \leq 
    \limsup_{n \to \infty}\frac{\bar{L}_{\tA}}{q} \nonumber\\
    & \leq
    \limsup_{n \to \infty} \frac{1}{r} \frac{n_1 n_2}{n^2}  \left(1-\frac{r}{K_1}\right)\nonumber\\
    &\quad+
    \limsup_{n \to \infty} \frac{1}{r} \frac{n^2_2}{n^2}  \left(1-\frac{r}{K_1}\right)\nonumber\\
    &\quad+
    \limsup_{n \to \infty} \frac{n_2 (n_1 - n_2)}{n^2} \nonumber\\
    & =
    \frac{1}{2r}\left(1-\frac{2r}{K}\right). 
\end{align}

\section{Converse for the Random Bi-partite Model}
\label{subsec:RB-converse}

Here we provide a lower bound on the optimal average communication load for the random bi-partite model that is within a constant factor of the upper bounds and complete the proof of Theorem \ref{thm:rndmbiprt}. Consider $\mathcal{G}=(\mathcal{V}_1 \cup \mathcal{V}_2,\mathcal{E})$ and assume that $n_1\geq n_2$. To derive a lower bound on $L^*(r)$, for every realization of RB$(n_1,n_2,q)$ graph, we arbitrarily remove $n_1 - n_2$ vertices in $\mathcal{V}_1$ along with their corresponding edges. The new bi-partite  graph represents two random ER graphs with $n_2$ vertices. Consider Reducing the vertices in one side of the new graph, e.g. $\mathcal{V}_2$. Clearly, this  provides a lower bound on $L^*(r)$. Note that now each Mapper can benefit from a redundancy factor of $2r$. According to Theorem \ref{thm:ERCDC}, Reducing  $\mathcal{V}_2$ induces the (optimal) communication load of $\frac{1}{2r} q \left(1 - \frac{2r}{K}\right) + o(q)$ which implies
\begin{align}
\label{eq:converseRB}
    \limsup_{n \to \infty}\frac{L^*(r)}{q}
    & \geq 
    \limsup_{n \to \infty}\frac{1}{2r} q \frac{n^2_2}{n^2} \left(1 - \frac{2r}{K}\right) + o(q) \nonumber \\
    & =
    \frac{1}{8r}\left(1-\frac{2r}{K}\right). 
\end{align}
Hence, the proof of converse of Theorem \ref{thm:rndmbiprt} is complete. Furthermore, (\ref{eq:achievbRB}) and (\ref{eq:converseRB}) together asymptotically characterize the optimal average normalized communication load $L^*(r)$ within a factor of $4$.

\section{Achievability for the Stochastic Block Model}\label{subsec:sbm}

In this Section, we specialize our proposed scheme in Section \ref{sec:achvblty} for the stochastic block model and prove the achievability of Theorem \ref{thm:sbm}. Consider an $\text{SBM}(n_1,n_2,p,q)$ graph $\mathcal{G}=(\mathcal{V}_1 \cup \mathcal{V}_2,\mathcal{E}_1 \cup \mathcal{E}_2 \cup \mathcal{E}_3)$ with $n=n_1+n_2$,  $|\mathcal{V}_1|=n_1=\Theta(n)$, and $|\mathcal{V}_2|=n_2=\Theta(n)$. Edge subsets $\mathcal{E}_1$, $\mathcal{E}_2$ and $\mathcal{E}_3$ respectively represent intra-cluster edges among vertices in $\mathcal{V}_1$, intra-cluster edges among vertices in $\mathcal{V}_2$, and inter-cluster edges between vertices in $\mathcal{V}_1$ and $\mathcal{V}_2$. Let $\mathcal{G}_1=(\mathcal{V}_1,\mathcal{E}_1)$ and $\mathcal{G}_2=(\mathcal{V}_2,\mathcal{E}_2)$ be graphs induced by $\mathcal{V}_1$ and $\mathcal{V}_2$, respectively, and denote the graph of inter-cluster connections by $\mathcal{G}_3=(\mathcal{V}_1 \cup \mathcal{V}_2,\mathcal{E}_3)$. Clearly, $\mathcal{G}_1$ and $\mathcal{G}_2$ are $\text{ER}(n_1,p)$ and $\text{ER}(n_2,p)$ graphs, while $\mathcal{G}_3$ is $\text{RB}(n_1,n_2,q)$ graph.

Subgraph and Reduce allocations are described as follows. Mappers in $\mathcal{V}_1$ and Reducers in $\mathcal{V}_2$ are distributedly allocated to $K$ servers according to the allocation scheme proposed in Section \ref{sec:achvblty}. Similarly, Mappers in $\mathcal{V}_2$ and Reducers in $\mathcal{V}_1$ are distributedly allocated to $K$ servers according to the allocation scheme proposed in Section \ref{sec:achvblty}. Therefore, each server Maps $n_1 r/K$ vertices in $\mathcal{V}_1$ and $n_2 r/K$ vertices in $\mathcal{V}_2$, inducing the computation load $r$. Moreover, each server Reduces $n_1/K$ functions in $\mathcal{V}_1$ and $n_2/K$ functions in $\mathcal{V}_2$. We consider this allocation, denoted by $\tA$, for both uncoded and coded Shuffling schemes. In uncoded scheme, Reducing each function in $\mathcal{V}_1$ requires on average $p n_1$ intermediate values Mapped by vertices in $\mathcal{V}_1$ due to intra-cluster connections which introduces the average uncoded load 
$
\bar{L}^{\textsf{UC1}}_{\tA} = p \frac{n^2_1}{(n_1 + n_2)^2}  \left(1-\frac{r}{K}\right).
$
Similarly, the average uncoded load for Reducing $\mathcal{V}_2$ due to intra-cluster connections
is 
$
\bar{L}^{\textsf{UC2}}_{\tA} = p \frac{n^2_2}{(n_1 + n_2)^2}  \left(1-\frac{r}{K}\right).
$
Moreover, inter-cluster connections induce an average load 
$
\bar{L}^{\textsf{UC3}}_{\tA} = q \frac{2 n_1 n_2}{(n_1 + n_2)^2} \left(1-\frac{r}{K}\right).
$

In the coded scheme, we propose to employ coded Shuffling for the ER and RB models in the regime of interest, that is $p = \omega(\frac{1}{n^2})$, $q = \omega(\frac{1}{n^2})$ and $p \geq q$. Thus, the overall communication load can be decomposed into three components. We first apply the coded Shuffling scheme described in Section \ref{sec:GCDC} to ER graph $\mathcal{G}_1$ which induces the average normalized communication load 
\begin{align}
    \bar{L}^{\textsf{C1}}_{\tA}\leq \frac{1}{r}\bar{L}^{\textsf{UC1}}_{\tA}  +o(p) = \frac{1}{r}p \frac{n^2_1}{(n_1 + n_2)^2}  \left(1-\frac{r}{K}\right)+o(p). \nonumber
\end{align}
Similarly, the same scheme applied to ER graph $\mathcal{G}_2$ results in the average normalized communication load 
\begin{align}
    \bar{L}^{\textsf{C2}}_{\tA}\leq \frac{1}{r}\bar{L}^{\textsf{UC2}}_{\tA}  +o(p) = \frac{1}{r}p \frac{n^2_2}{(n_1 + n_2)^2} \left(1-\frac{r}{K}\right)+o(p).  \nonumber
\end{align}
Finally, we employ the same scheme twice for the two ER models constituting the RB graph $\mathcal{G}_3$ which induces the average normalized communication load  
\begin{align}
    \bar{L}^{\textsf{C3}}_{\tA}\leq
    \frac{1}{r}\bar{L}^{\textsf{UC3}}_{\tA}  +o(q) = \frac{1}{r}q \frac{2 n_1 n_2}{(n_1 + n_2)^2}  \left(1-\frac{r}{K}\right)+o(q). \nonumber
\end{align}
Let us denote by $ \bar{L}^\textsf{C}_{\tA}$ and $ \bar{L}^\textsf{UC}_{\tA}$ the total average normalized communication loads of the coded and uncoded schemes, respectively. Therefore, 
\begin{align}
    L^*(r) &\leq \bar{L}^{\textsf{C}}_{\tA} \nonumber\\
    &= \bar{L}^{\textsf{C1}}_{\tA} + \bar{L}^{\textsf{C2}}_{\tA}  + \bar{L}^{\textsf{C3}}_{\tA} \nonumber\\
    & \leq \frac{1}{r}(\bar{L}^{\textsf{UC1}}_{\tA} + \bar{L}^{\textsf{UC2}}_{\tA}  + \bar{L}^{\textsf{UC3}}_{\tA}) + o(p)  \nonumber\\
    &=\frac{1}{r}\bar{L}^{\textsf{UC}}_{\tA}+ o(p)  \nonumber\\
    &= \frac{p n_1^2+ p n_2^2+2q n_1 n_2}{(n_1+n_2)^2}\left(1-\frac{r}{K}\right)+ o(p), \nonumber
\end{align}
which concludes the proof of achievability of Theorem \ref{thm:sbm}.

\section{Converse for the Stochastic Block Model}
\label{subsec:SBM-converse}

In this section, we provide the proof of the converse of Theorem \ref{thm:sbm}. Consider an $\text{SBM}(n_1,n_2,p,q)$ graph $\mathcal{G}=(\mathcal{V}_1 \cup \mathcal{V}_2,\mathcal{E}_1 \cup \mathcal{E}_2 \cup \mathcal{E}_3)$ with $n=n_1+n_2$,  $|\mathcal{V}_1|=n_1=\Theta(n)$, and $|\mathcal{V}_2|=n_2=\Theta(n)$. Our approach to derive a lower bound for the minimum average communication load is to randomly remove edges from the two intra-cluster edges, i.e. $\mathcal{E}_1$ and $\mathcal{E}_2$. Moreover, edges are removed such that each of those clusters are then Erdos-Renyi models with connectivity probability $q$ (reduced from $p$). This can be simply verified by the following coupling-type argument. Let the Bernoulli random variable $E_p$ denote the indicator of existence of a generic edge in an $\text{ER}(n,p)$ graph, i.e. $\Pr[E_p=1] = 1-p$. Now, generate another Bernoulli $E_q$ by randomly removing edges from the realized ER graph as follows:
\begin{align}
    E_q =
    \left\{
	\begin{array}{ll}
		\mbox{if } E_p = 0  & 0 \\
		 \mbox{if } E_p = 1 & \left\{
	\begin{array}{ll}
		0  & \mbox{w.p. } 1-q/p \\
		1 & \mbox{w.p. } q/p.
	\end{array}
    \right.
	\end{array}
    \right. \nonumber
\end{align}
Clearly, $E_q$ is Bernoulli$(q)$ and the resulting graph has fewer number of edges compared to the original one (with probability $1$). By doing so for the two ER components of the SBM graph, we have a larger ER graph of size $n=n_1+n_2$ with connectivity probability $q$. Using the converse in Theorem \ref{thm:ERCDC}, we have the following for average normalized communication load for the stochastic block model:
\begin{align}
    \frac{L^*(r)}{q} 
    & \geq
    \frac{1}{r} \left(1-\frac{r}{K}\right). \nonumber
\end{align}

\section{}
\label{sub:applemma}
\begin{lemma}\label{lemma:ineq}
For all $p\in [0,1]$ and $s'>0$, we have $\big(p e^{s'}+1-p \big)^2 \leq p e^{2s'}+1-p$.
\end{lemma}
\begin{proof}
 For given $p\in [0,1]$, define $f(s')=\big(p e^{s'}+1-p \big)^2 - \big(p e^{2s'}+1-p\big)$. Clearly $f(0)=0$. Moreover,
 \begin{equation}
      f'(s')=2p\bar{p} (e^{s'}-e^{2s'}) < 0, \nonumber
  \end{equation}
  for $s'>0$. Therefore, $f(s')\leq 0$ for all $s'>0$, concluding the claim of the lemma.
\end{proof}

\end{appendices}

\begin{IEEEbiographynophoto}{Saurav Prakash} 
received the Bachelor of Technology degree in Electrical Engineering from the Indian Institute of Technology (IIT), Kanpur, India in 2016. He is currently pursuing the Ph.D. degree in Electrical and Computer Engineering at the University of Southern California (USC), Los Angeles. He was a finalist in the Qualcomm Innovation Fellowship program in 2019. His research interests include information theory and data analytics with applications in large-scale machine learning and edge computing. He received the Annenberg Graduate Fellowship in 2016 and was one of the Viterbi-India fellows in summer 2015. 
\end{IEEEbiographynophoto}

\begin{IEEEbiographynophoto}{Amirhossein Reisizadeh}
received his B.S. degree form Sharif University of Technology, Tehran, Iran in 2014 and an M.S. degree from University of California, Los Angeles (UCLA) in 2016, both in Electrical Engineering. He is currently pursuing his Ph.D. in  Electrical and Computer Engineering at University of California, Santa Barbara (UCSB). He was a finalist in the Qualcomm Innovation Fellowship program in 2019. He is interested in using information and coding-theoretic concepts to develop fast and efficient algorithms for large-scale machine learning, distributed computing and optimization.
\end{IEEEbiographynophoto}

\begin{IEEEbiographynophoto}{Ramtin Pedarsani}
is an Assistant Professor in ECE Department at the University of California, Santa Barbara. He received the B.Sc. degree in electrical engineering from the University of Tehran, Tehran, Iran, in 2009, the M.Sc. degree in communication systems from the Swiss Federal Institute of Technology (EPFL), Lausanne, Switzerland, in 2011, and his Ph.D. from the University of California, Berkeley, in 2015. His research interests include machine learning, information and coding theory, networks, and transportation systems. Ramtin is a recipient of the IEEE international conference on communications (ICC) best paper award in 2014.
\end{IEEEbiographynophoto}

\begin{IEEEbiographynophoto}{A. Salman Avestimehr}
is a Professor and director of the Information Theory and Machine Learning (vITAL) research lab at the Electrical and Computer Engineering Department of University of Southern California. He received his Ph.D. in 2008 and M.S. degree in 2005 in Electrical Engineering and Computer Science, both from the University of California, Berkeley. Prior to that, he obtained his B.S. in Electrical Engineering from Sharif University of Technology in 2003. His research interests include information theory, coding theory, and large-scale distributed computing and machine learning.

Dr. Avestimehr has received a number of awards for his research, including the James L. Massey Research \& Teaching Award from IEEE Information Theory Society, an Information Theory Society and Communication Society Joint Paper Award, a Presidential Early Career Award for Scientists and Engineers (PECASE) from the White House, a Young Investigator Program (YIP) award from the U. S. Air Force Office of Scientific Research, a National Science Foundation CAREER award, the David J. Sakrison Memorial Prize, and several Best Paper Awards at Conferences. He is a Fellow of IEEE. He has been an Associate Editor for IEEE Transactions on Information Theory. He is currently a general Co-Chair of the 2020 International Symposium on Information Theory (ISIT).
\end{IEEEbiographynophoto}

\end{document}